\documentclass[12pt]{amsart}

\usepackage[margin=2.5cm]{geometry}
\usepackage{tikz}
\tikzset{node distance=2cm, auto}
\usetikzlibrary{matrix,arrows,decorations.pathmorphing}

\usepackage{graphicx}
\usepackage{amsmath,amssymb}

\usepackage{amscd}
\usepackage{verbatim}
\usepackage{enumerate}

\newtheorem{theorem}{Theorem}
\newtheorem{definition}{Definition}
\newtheorem{proposition}{Proposition}
\newtheorem{lemma}{Lemma}

\newtheorem{corollary}{Corollary}

\newcommand{\g}{\mathfrak{g}}

\newcommand{\hfr}{\mathfrak{h}}

\newcommand{\RR}{\mathbb{R}}
\newcommand{\CC}{\mathbb{C}}

\newcommand{\CH}{\mathcal{H}}
\newcommand{\CL}{\mathcal{L}}
\newcommand{\CO}{\mathcal{O}}

\newcommand{\CR}{\mathcal{R}}

\newcommand{\Diff}{\mathbb{D}}
\newcommand{\Vect}{\mbox{Vect}}
\newcommand{\End}{\textup{End}}
\newcommand{\R}{\mathcal{R}}

\newcommand{\pa}{\partial}
\usepackage[colorlinks=true]{hyperref}
\usepackage{stmaryrd} % for \sslash
\usepackage{commath} % for \pd to type partial derivatives faster
\usepackage{empheq}

\newenvironment{dedication}
        {\vspace{6ex}\begin{quotation}\begin{center}\begin{em}}
        {\par\end{em}\end{center}\end{quotation}}

\title{Spin Calogero-Moser periodic chains and two dimensional Yang-Mills theory with corners.}

\author{Nicolai Reshetikhin}
\address{N.R.: YMSC, Tsinghua University, Beijing, China,\& BIMSA, Beijing, China \& Department of Mathematics, University of California, Berkeley, CA 94720, USA.}
\email{reshetik@math.berkeley.edu}

\usepackage{fancyhdr}
\begin{document}

\maketitle

\begin{dedication}
\hspace{4cm}
\vspace*{3cm}{Dedicated to the memory of Vaughan Jones.}
\end{dedication}

\begin{abstract}
Quantum Calogero-Moser spin system is a superintegable system 
with the spectrum of commuting Hamiltonians that can be described 
entirely in terms of representation theory of corresponding simple Lie group.
In this paper the underlying Lie group $G$ is a compact connected, simply connected simple Lie group. It has a natural generalization known as quantum Calogero-Moser spin chain. In the first part of the paper we show that quantum Calogero-Moser spin chain is a quantum
superintegrable systems. Then we show that the Euclidean multi-time propagator 
for this model can be written as 
a partition function of a two-dimensional Yang-Mills theory on a cylinder.
Then we argue that the two-dimensional Yang-Mills theory with Wilson loops 
with "outer ends" should be regarded as the theory on space times with non-removable
corners. Partition functions of such theory satisfy non-stationary Calogero-Moser equations.
\end{abstract}

\section*{Introduction} 

In this paper we consider the two-dimensional Yang-Mills (2D YM) theory on surfaces with 
open Wilson graphs. The partition function in this case is a vector in the tensor product of
finite dimensional representations of a simple compact Lie group. 
We show that it is a solution to $N$-point spin Calogero-Moser (CM) 
evolution with initial conditions 
determined by the structure of the Wilson graph. 

One can draw an analogy between these models and the Chern-Simons theory 
with Wilson lines ending on the two-dimensional boundary of a 3-manifold. 

Quantum spin Calogero-Moser systems have a long history. 
See for example \cite{Feh}\cite{R1} and
references therein. In this paper we will work with so called $N$-spin version of CM systems. 
These models of many-particle quantum systems with internal (spin) degrees of 
freedom are examples of quantum superintegrable systems. For the superintegrability of quantum CM spin systems and for the notion of quantum superintegrability see \cite{R1}\cite{R-ICMP}. For the superintegrability of classical $N$-point spin CM chain see \cite{R4}

Two dimensional model of Yang-Mills (YM) theory is a two dimensional topological 
quantum field theory. A non-perturbative partition function for a disc was introduced in \cite{Mig}. The 2D YM theory was studied as a TQFT in papers
\cite{Witt-1}\cite{Witt-2} where among other things a convincing argument
was given that the non-perturbative formula is a version of the Duijstermaat-Heckman
localization. For the latest study of perturbative aspects of the YM theory on surfaces with boundary
and for more references see \cite{Mn}. To be more precise, the YM theory on surfaces is not exactly topological.
The partition function depends on the total area of the surface (assuming it is finite), and, possibly,
on more parameters which are additive with respect to gluing, when the rank of the gauge group is
greater than one \cite{Witt-2}.

The partition function and expectation values of Wilson graphs (line observables 
in this TQFT) can be interpreted as an analytic continuation of the semiclassical limit
of quantum Chern-Simons invariants \cite{PR}\cite{Witt-2}. If the
corresponding quantum group \cite{RT} has $q=\exp(i\frac{\pi}{r})$ with $r=k+h^\vee$ where $k$ is the level of the quantum Chern-Simons theory \cite{W}, $h^\vee$ is the dual Coxeter number, the semiclassical limit is when $k\to \infty$. The quantum Chern-Simons invariant of
of a circle bundle over $\Sigma$ with the Chern class $m$ in the limit when $k\to\infty, m\to\infty$ and $a=\frac{m}{r}$ is finite becomes the  partition function of the 2D YM with the area $ia$ in the limit . When $m=0$ this was first observed in \cite{Witt-2}. For $a\neq 0$
see \cite{PR}.

In this paper we show that the partition function of the 2D YM theory on surfaces with open Wilson graphs are solutions to non-stationary Calogero-Moser type differential equations. For example the partition function on a cylinder with Wilson lines that are parallel to the cylinder is the multi-time propagator for quantum $N$-spin Calogero-Moser system with Euclidean time being areas of corresponding strips of the cylinder.
We also argue that such partition functions can be interpreted as a $2-1-0$ TQFT with non-removable corners (points at the boundary where Wilson lines "go out"). 

In the first section we define quantum $N$-spin Calogero-Moser model for a
compact simple Lie group $G$. In the second section we explain the relation between quantum $N$-spin Calogero-Moser system and the 2D YM with open Wilson graphs. 

\subsection{Acknowledments} This paper was started as a joint project with Jasper Stokman. The author is grateful to Jasper for many discussions and for the collaboration on this paper.
He also would like to thank S. Gukov, V. Kazakov, P. Mnev and G. Moore for references, discussions and comments and to D. Solovyev for help with figures. 
N.R. was partially supported by 
the Collaboration Grant "Categorical Symmetries" from the Simons 
Foundation, by the NSF grant DMS-1902226 and by the RSF grant 21-11-00141. 
This paper grew out of the talk given by N.R. at the "Physics and Geometry" 
seminar at BIMSA in the Spring 2022. He is grateful to participants of this seminar for 
interesting discussions.

\section{Quantum $N$-spin Calogero-Moser system and its superintegrability}

\subsection{Quantum superintegrable systems}
Let $A_h$ be a family of associative algebras with trivial center. We think of each algebra $A$ from this family as the algebra of quantum observables for a quantum system. 

A commutative algebra has rank $k$ means that it has maximum $k$ algebraically independent elements, i.e. $k$ is the Krull dimension of $I$. 
We will say that an associative algebra $A$ has rank $n$ if $n$ is the rank of maximal commutative subalgebra in $A$. 

\begin{definition} A commutative subalgebra $I\subset A$ of rank $k$ is a quantum superintegrable
system on $A$ if the rank of the centralizer of $I$ in $A$ is $n$ where $n$ is the rank of $A$.
\end{definition}

For a superintegrable quantum system we have embeddings of associative algebras:
\[
I\subset Z(I,A)\subset A
\]
where $Z(I,A)$ is the centralizer of $I$ in $A$ and $\mbox{rank}(Z(I, A))=\mbox{rank}(A)$.

An example of a superintegrable system is when $A$ is the algebra of differential operators in variable $x_1,\dots, x_n$ with polynomial coefficients, $I$ is the algebra of differential operators in $x_1,\dots, x_k$ with constant coefficients. In this case $Z(I,A)$ is the algebra of differential operators in $x_1,\dots, x_n$ with polynomial coefficients in $x_{k+1},\dots, x_n$.

\subsection{Algebraic preliminaries}
In this section we develop an algebraic set up for defining quantum $N$-spin Calogero-Moser system.

\subsubsection{} Throughout this paper $G$ is a compact, simple, simply connected Lie group. We will  denote by $\mathcal{R}(G)$ the algebra of functions on $G$ generated by matrix elements of finite
dimensional representations. It is a span of matrix elements of irreducible representations\footnote{ There is a natural isomorphism of $G\times G$-modules 
\[
\CR(G)\simeq \bigoplus_\lambda V_\lambda\otimes V_\lambda^*
\]
where the sum is taken over all irreducible finite dimensional representation of
$G$ and $V_\lambda$ is the irreducible representation with the highest weight $\lambda$.
Note that space on the right consists of finite direct sums of unbounded length.}.

The space $\CR(G)$ is an open dense subspace in the Hilbert space $L^2(G)$ of square 
integrable functions with respect the Haar measure on $G$. 

We denote by $\Diff(G)$ the algebra of differential operators on $G$ preserving $\CR(G)$. Let $\Vect(G)$ be the subalgebra 
in the Lie algebra of all vector fields on $G$ which correspond to first order differential
operators in  $\Diff(G)$.

The Lie group $G$ acts on itself by left and right translations: $h: g\mapsto hg$ and $h: g\mapsto gh^{-1}$. These actions give two homomorphisms of Lie algebras $\g\to \Vect(G)$ and, as a consequence, two algebra inclusions $\widehat{\mu}_\ell,\widehat{\mu}_r: U\g\rightarrow\Diff(G)$ (quantum mement maps) as right and left $G$-invariant differential operators, respectively. 

Let $Z(\g)$ be the center of $U\g$, and $S$ the antipode for $U\g$, i.e. 
the anti-involution of $U\g$ that acts as $x\to -x$ on elements $\g\subset U\g$.

\subsubsection{} Consider the group $G^N$ acting on itself by left translations :
\[
(h_1,\dots, h_N): (g_1,g_2,\dots, g_N)\mapsto (h_1g_1, h_2g_2, \dots, h_{N}g_{N})
\]
by right translations twisted by the cyclic permutation:
\[
(h_1,\dots, h_N): (g_1,g_2,\dots, g_N)\mapsto (g_1h_2^{-1}, g_2h_3^{-1}, \dots, g_{N}h_1^{-1})
\]

The combination of these two actions
acts as:
\begin{equation}\label{g-action}
(h_1,\dots, h_N): (g_1,g_2,\dots, g_N)\mapsto (h_1g_1h_2^{-1}, h_2g_2h_3^{-1}, \dots, h_{N}g_{N}h_1^{-1})
\end{equation}

Note that the action by twisted conjugations has a natural interpretation as 
the gauge action parallel transports for a connection in a trivial principal $G$-bundle over a circle with with marked points. In this interpretation
$g_i$ is the holonomy along the interval connecting points $i$ and $i+1$, and
$h_i$ is the gauge transformation at the $i$-th point, see Fig. \ref{circle}. This is 
why we {\it call (\ref{g-action}) the gauge action}.

\begin{figure}[t!]\label{circle}
\begin{center}
\includegraphics[width=0.6\textwidth, angle=0.0, scale=0.7]{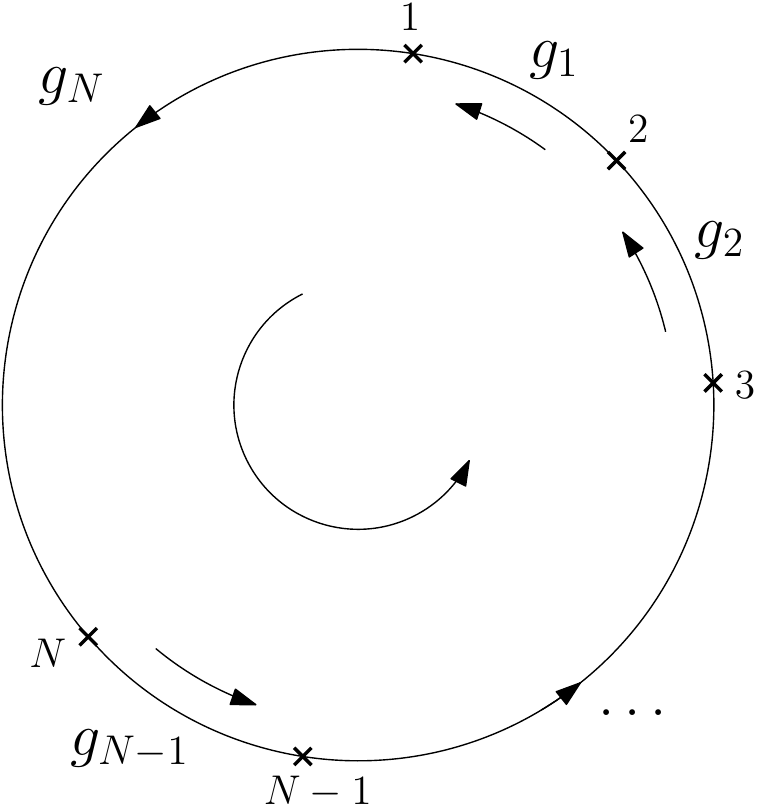}
\caption{\label{circle} Holonomies along integravls and gauge transformations. Here $g_i$ stands for the holonomy between points $i$ and $i+1$. The gauge transformation at points $i-1$ and $i$ acts on $g_i$ as $g_i\mapsto h_{i}g_ih_{i+1}^{-1}$. Arrow show the orientation of the circle and of the disc that it bounds. }
\end{center}
\end{figure}

The left, the right and the gauge action of $G^N$ on itself lift to actions on $\Diff(G^N)$.
Quantum moment maps for these actions give injective homomorphisms of associative algebras
$\widehat{\mu}_L, \widehat{\mu}_R, \widehat{\mu}_{ad} : U\g^{\otimes N}\subset \Diff(G)$.

This gives the following chain of inclusions of associative algebras:
\[
\Diff(G^N)\stackrel{\widehat{\mu}_L\otimes \widehat{\mu}_R}{\hookleftarrow} (U\g)^{\otimes N} \otimes (U\g)^{\otimes N}
\hookleftarrow (U\g)^{\otimes N}
\]
Here the right inclusion is simply mapping to the left factor in the tensor product\footnote{ We could have taken the inclusion to the right factor. At the end it does make a difference.}

This maps filter through algebra homomorphisms:
\[
\Diff(G^N)\stackrel{\widehat{\mu}_L\otimes \widehat{\mu}_R}{\hookleftarrow} (U\g)^{\otimes N} \otimes_{Z(\g)^{\otimes N}} (U\g)^{\otimes N}
\hookleftarrow (U\g)^{\otimes N}
\]
Here $Z(\g)\subset U\g$ is the center of the universal enveloping algebra and 
$\otimes_{Z(\g)^{\otimes N}}$ is the quotient of the algebraic tensor product. In this tensor
product over $Z(\g)^{\otimes N}$ we have:
\begin{eqnarray*}
(a_1z_1\otimes a_2z_2\dots\otimes a_Nz_N)&\otimes_{Z(\g)^{\otimes N}}&(b_1\otimes b_2\dots\otimes b_N)=\\
(a_1\otimes a_2\dots\otimes a_N)&\otimes_{Z(\g)^{\otimes N}}&(S(z_1)b_1\otimes S(z_2)b_2\dots\otimes S(z_{N})b_N).
\end{eqnarray*}
Here $S(z)$ is the antipode of the central element $z$. The antipode in this case is an algebra antiautomorphism of $U(\g)$ which acts as $S(x)-x$ on linear elements $x\in \g\subset U(\g)$.

The last embedding is $a\mapsto a\otimes_{Z(\g)^{\otimes N}} 1$. Note that we have a natural isomorphism 
\[
(U\g)^{\otimes N} \otimes_{Z(\g)^{\otimes N}} (U\g)^{\otimes N}\simeq ((U\g) \otimes_{Z(\g)} (U\g))^{\otimes N}
\]
The action of the gauge group $G^N$ extends naturally to its action on $\Diff(G)$. It acts by the adjoint action on each copy of $U(\g)^{\otimes N}$ and it acts diagonally on $(U\g)^{\otimes N} \otimes_{Z(\g)^{\otimes N}} (U\g)^{\otimes N}$. The $G^N$-invariant part of these embeddings is:

\begin{equation}\label{incl-1}
\Diff(G^N)^{G^N}\stackrel{\widehat{\mu}_L\widehat{\otimes} \widehat{\mu}_R}{\hookleftarrow} (U\g^{\otimes N} \otimes_{Z(\g)^{\otimes N}} U\g^{\otimes N})^{G^N}
\hookleftarrow Z(\g)^{\otimes N}
\end{equation}

Here $Z(\g)^{\otimes N}$ is the $G^N$-invariant part of $U\g^{\otimes N}$ with respect to the adjoint action.

\subsection{Quantum $N$-spin Calogero-Moser system}

\subsubsection{The space of states} Let $V_{\mu_1}, \dots, V_{\mu_N}$ be finite dimensional irreducible representations
of $G$ with highest weights $\mu_i$ and representation maps $\pi^{\mu_i}$. 
We will use the same notation for associated irreducible representations of $U\g$, $\pi^{\mu_i}: U\g\rightarrow\End(V_{\mu_i})$.
Fix Hermitian scalar product $(\cdot,\cdot)_{\mu_i}$
on each $V_{\mu_i}$ turning it into a unitary representation. Denote $A^\dag$ the Hermitian  conjugate to $A\in \End(V_{\mu_i})$ relative to $(\cdot,\cdot)_{\mu_i}$. The unitarity  of the representation $V_{\mu_i}$ can be written as 
\[
\pi_{\mu_i}(g)^\dag=\pi_{\mu_i}(g^{-1}).
\]
Now let us define the representation space for $\Diff(G^N)^{G^N}$ corresponding to representations $\mu_i$,

Consider the following space of gauge group equivariant $V_{\mu_N}\otimes\cdots\otimes V_{\mu_1}$-valued functions of $G^N$:
\[
\R_{\mu}=\R(G^N\to V_{\mu_1}\otimes\cdots\otimes V_{\mu_N}):= \bigl(\R(G^N)\otimes V_{\mu_1}\otimes\cdots\otimes V_{\mu_N}\bigr)^{G^N}.
%L_2(G^N\to V_{\mu_N}\otimes \dots\otimes V_{\mu_1})^{G^N}
\]
were $\R(G^N)$ is the span of matrix elements of all finite dimensional representations of $G^N$.
The $G^N$-equivariance condition for $f\in\R_\mu$ is
\begin{equation}\label{eqprop}
f(g^h)=(\pi^{\mu_1}(h_1)\otimes\cdots\otimes\pi^{\mu_N}(h_N))f(g_1,\ldots,g_N)
\end{equation}
where $g^h$ is the gauge action of $h\in G^N$ on $g\in G^N$ .

Clearly this is a module over $\Diff(G^N)$ with natural action of differential 
operators on functions.

This space has a Hermitian scalar product
\[
(f_1,f_2)=\int_{G^N}(f_1(g_1,\dots, g_N),f_2(g_1,\dots, g_N))_\mu dg_1\dots dg_N.
\]
Here $(v_1\otimes\dots\otimes v_N,u_1\otimes\dots\otimes u_N)_\mu=\prod_{i=1}^N(v_i,u_i)_{\mu_i}$ and $dg$ is the Haar measure on $G$ normalized such that $\int_Gdg=1$.

The completion $\CH_\mu$ of $\R_\mu$ with respect to the norm $||f||^2=(f,f)$ is the gauge equivariant subspace in $V_{\mu_1}\otimes\cdots\otimes V_{\mu_N}$-valued square integrable functions
on $G^N$ satisfying \eqref{eqprop}. This is the {\it space of states} of quantum $N$-spin Calogero-Moser system.

\subsubsection{Quantum moment maps} Let $X\in \g\subset U\g$, define $X_j=1^{\otimes (j-1)}\otimes X \otimes 1^{\otimes (N-j-1)}\in (U\g)^{\otimes N}$. The images of $X_j$ in $\Diff(G^N)$ with respect to the left and to the right quantum moment map act on $\R(G^N)\otimes V_{\mu_N}\otimes\cdots\otimes V_{\mu_1}$ as follows
\[
\widehat{\mu}_R(X_j)f(g_1,\dots, g_j, g_{j+1}, \dots, g_N)=\frac{d}{dt}f(g_1,\dots, g_je^{tX}, \dots, g_N)|_{t=0},
\]
where $X\mapsto e^X$ is the exponential map of $G$. 
\[
\widehat{\mu}_L(X_j)f(g_1,\dots, g_j, g_{j+1}, \dots, g_N)=\frac{d}{dt}f(g_1,\dots, e^{-tX}g_j, \dots, g_N)|_{t=0}.
\]
The image of $X_j$ with respect to the quantum moment map for
the gauge action $\widehat{\mu}(X_j)=\widehat{\mu}_L(X_j)+\widehat{\mu}_R(X_{j+1})$ acts as
\[
\widehat{\mu}(X_j)f(g_1,\dots, g_{j-1}, g_{j}, \dots, g_N)=\frac{d}{dt}f(g_1,\dots, g_{ij1}e^{-tX}, e^{tX}g_{j}, \dots, g_N)|_{t=0}
\]

Gauge innvariance (\ref{eqprop}) implies that for each $f\in\R_\mu$, we have
\begin{equation}\label{f-prop}
\widehat{\mu}(X_j)f(g_1,\ldots,g_N)=\pi_j^{\mu_j}(X)f(g_1,\ldots,g_N)
\end{equation}
where $\pi_j^{\mu_j}=\textup{id}_{V_{\mu_N}}\otimes\cdots\otimes\textup{id}_{V_{\mu_{j+1}}}\otimes\pi^{\mu_j}\otimes\textup{id}_{V_{\mu_{j-1}}}\otimes\cdots\otimes
\textup{id}_{V_{\mu_1}}$.

\subsubsection{Quantum Hamiltonians and quantum integrals} Recall that $I_\mu$ is the subalgebra of quantum Hamiltonians and $J_\mu$ is the subalgebra of quantum integrals.

Because $Z(\g)^{\otimes N}$, $\bigl((U\g\otimes_{Z(\g)}\otimes U\g)^{\otimes N}\bigr)^{G^N}$ are subalgebras of $\Diff(G^N)^{G^N}$, they act on the space $\CR_\mu(G^G)$. Denote the images of the subalgebras $Z(\g)^{\otimes N}$, $\bigl((U\g\otimes_{Z(\g)}\otimes U\g)^{\otimes N}\bigr)^{G^N}$ and  in $\textup{End}(\R_\mu)$
by $I_\mu$, $J_\mu$ and $A_\mu$. They form a chain of subalgebras

\begin{equation}\label{s-int-alg}
A_\mu\hookleftarrow J_\mu\hookleftarrow I_\mu
\end{equation}

The center $Z(\g)$ as commutative algebra is a polynomial ring in $r=\textup{rank}(\g)$ homogeneous generators $c_{d_1},\ldots,c_{d_r}$ of degrees $d_1=2\leq d_2\leq\cdots\leq d_r$ respectively\footnote{Numbers $d_i-1, 1\leq i\leq r$ are called the exponents of $\mathfrak{g}$.}.  Note that $c:=c_2\in Z(\g)$ is the quadratic Casimir element determined by the Killing form.

Denote by 
\[
H^{\mu,(j)}_{d_k}\in I_\mu, \qquad\quad 1\leq j\leq N,\, 1\leq k\leq r
\]
the action of $c^{(j)}_k:=1^{\otimes (j-1)}\otimes c_k\otimes 1^{\otimes (N-j-1)}\in Z(\g)^{\otimes N}$ on $\R_\mu$. The Hamiltonians $H^{\mu,(j)}_{d_k}$ commute and generate the algebra $I_\mu$.

We will compute the radial components of the quadratic Hamiltonians $H^{(j)}:=H^{(j)}_2$ ($1\leq j\leq N$) in Subsection \ref{SectionGauge}.

\subsubsection{The spectrum of quantum Hamiltonians} 
Fix $N$ finite dimensional irreducible $G$-modules  $V_{\nu_1},\ldots,V_{\nu_N}$. Denote by $\R^{(\nu)}(G^N)$ the subspace of $\R(G^N)$ spanned by matrix elements of the irreducible $G^N$-representation $V_{\nu_1}\otimes\cdots\otimes V_{\nu_N}$\footnote{It does not matter which basis we use to define matrix elements, the subspace will be the same.}. Then by Peter-Weyl theorem we have the decomposition with respect to the left and the right 
action of $G^N$:
\begin{equation}\label{RG}
\R(G^N)=\bigoplus_\nu\R^{(\nu)}(G^N)
\end{equation}
Here $\nu$ runs over the set of $N$-tuples of dominant integral weights. 

Subspaces $\R^{(\nu)}(G^N)\subset \R(G^N)$ are invariant with respect to the gauge action of
$G^N$, thus we can define subspaces 
\[
\R^{(\nu)}_\mu:=(\R^{(\nu)}(G^N)\otimes V_{\mu_1}\otimes\cdots\otimes V_{\mu_N})^{G^N} \subset \R_\mu,
\]
Note that this subspace is not empty only if $V_{\nu_i}\subset  V_{\nu_{i-1}}\otimes V_{\mu_i}$,  cyclically for $i=1,\dots, N$.

Because $Z(\g)$ acts on an irreducible representation  by the multiplication on the corresponding central character, the subspace $\R^{(\nu)}_\mu$ is an eigensubspace
for $I_\mu$ with 
\[
H_{d_k}^{(j)}f=c_{d_k}(\nu_j)f
\]
for any $f\in \R^{(\nu)}_\mu$. Here $c_{d_k}(\eta)$ is the value of $c_{d_k}\in Z(\g)$
on the irreducible finite dimensional module $V_\eta$.

The Peter-Weyl theorem for $\R(G^N)^{\otimes N}$ gives the decomposition of $\R_\mu$ in simultaneous eigensubpaces for the action of the Hamiltonians $H^{(j)}_{d_k}$:
\begin{equation}\label{decomptwist}
\R_\mu=\bigoplus_\nu\R^{(\nu)}_\mu
\end{equation}
 
For intertwiners $a_i\in\textup{Hom}_G(V_{\nu_i},V_{\nu_{i-1}}\otimes V_{\mu_i})$ (here $\nu_o=\nu_N$), define the {\it trace function} 
$\Psi^{(\nu)}_{a,\mu}\in\R^{(\nu)}(G^N)\otimes V_{\mu_1}\otimes\cdots\otimes V_{\mu_N}$ on $G^N$ as 

\begin{equation}\label{tr-fncn}
\begin{split}
\Psi^{(\nu)}_{a,\mu}(g_1,&\dots, g_N)\\
=&\textup{Tr}_{V_{\nu_1}}\bigl((a_1\pi^{\nu_1}(g_1)\otimes\textup{id}_{\textup{V}_{\mu_2}\otimes\cdots\otimes V_{\mu_N}})
\cdots (a_{N-1}\pi^{\nu_{N-1}}(g_{N-1})\otimes\textup{id}_{\textup{V}_{\mu_{N}}})a_{N}\pi^{\nu_N}(g_N)\bigr).
\end{split}
\end{equation}

It is clear that $\Psi^{(\nu)}_{a,\mu}$ is gauge invariant and that $\Psi^{(\nu)}_{a,\mu}\in\R^{(\nu)}_\mu$. 
The restriction of $\Psi^{(\nu)}_{a,\mu}$ to $1^{\times (N-1)}\times H$ gives generalized  trace functions from \cite{Et}.

\begin{theorem}
Functions $\Psi^{(\nu)}_{a,\mu}$ form a linear basis in $\R_\mu^{(\nu)}$ enumerated by a basis in $\otimes_{i=1}^N\textup{Hom}_G(V_{\nu_i},V_{\nu_{i-1}}\otimes V_{\mu_i})$.
\end{theorem}

\begin{corollary} We have a linear isomorphism
\[
\R_\mu^{(\nu)}(G^N)\simeq \otimes_{i=1}^N\textup{Hom}_G(V_{\nu_i},V_{\nu_{i-1}}\otimes V_{\mu_i})
\]
\end{corollary}
%By \eqref{decompelemtwisted}, the subspace $\R^{(\nu)}_{\mu}$ of $\R_\mu$ are invariant for the action $C_{A_\mu}(I_\mu)$. In particular, 

Let $Z(I_\mu, A_\mu)$ be the centralizer of the commutative subalgebra $I_\mu$ in $A_\mu$. By the construction of $J_\mu\subset A_\mu$ it is a subalgebra in the centralizer of $I_\mu$, i.e. $J_\mu\subset Z(I_\mu, A_\mu)$. Since 
$\R^{(\nu)}_\mu\subset \R_\mu$ is an eigensubspace of $I_\mu$ the centralizer 
$Z(I_\mu, A_\mu)$ acts on it. Thus, $\R^{(\nu)}_\mu$ is an $J_\mu$-module.

\begin{theorem}
The decomposition \eqref{decomptwist} is the multiplicitly free decomposition of $\R_\mu$ in irreducible $J_\mu$-modules.
\end{theorem}

The proof of this theorem and its generalization to a more general representation theoretical context will be given in \cite{RSQ}.

Denote by $J_\mu^{(\nu)}$ and $Z(I_\mu, A_\mu)^{(\nu)}$ images of $J_\mu$  and $Z(I_\mu, A_\mu)$
in $\End(\CR_\mu^{(\nu)})$ respectively.

\begin{corollary}
We have 
\[
J_\mu^{(\nu)}=Z(I_\mu, A_\mu)^{(\nu)},
\]
\end{corollary}

Thus, the embeddings (\ref{s-int-alg}) define a simple quantum superintegrable system in the sense of  \cite{R-ICMP}.

Using well known identities for integrals of matrix elements of irreducible representations
with respect to the Haar measure\footnote{All necessary identities are summarized in the 
Appendix \ref{A}} it is easy to prove the following formula for the scalar product of trace functions:
\[
\int_{G^N} (\Psi^{\{\mu\}}_{b,\lambda'}(g_1,\dots, g_N), \Psi^{\{\mu\}}_{a,\lambda}(g_1,\dots, g_N))  dg_1\dots dg_N=\prod_{i=1}^N \delta_{\lambda_i, \lambda_i'}\prod_{i=1}^N (b_i,a_i)_{\lambda_i}
\]
Here $(b, a)_{\lambda_i}$ is a Hermitian scalar product on $Hom_\g(V_{\lambda_i},  V_{\lambda_{i-1}}\otimes V_{\mu_i})$ defined as follows. Because all representations are equipped with Hermitian
structure for each $a\in Hom_\g(V_{\lambda_i}, V_{\lambda_{i-1}}\otimes V_{\mu_i})$ we have
its Hermitian conjugate $a^\dag: V_{\lambda_{i-1}}\otimes V_{\mu_i}\to V_{\lambda_i}$. For any such $a,b$, by Schur's lemma the composition $ab^\dag$ is a multiple of the identity $id_{V_{\lambda_i}}$. This defines the scalar product $(b,a)_{\lambda_i}$ as $ab^\dag=(b,a)_{\lambda_i}id_{V_{\lambda_i}}$.

\subsection{The gauge fixing}\label{SectionGauge}

There $N$ natural ways to identify the coset $G^N/G^N$ for the gauge action (\ref{g-action}) with the coset
$G/G$ with respect to the conjugation. 

Indeed, fix $i\in1,\dots, N$ and for $j\neq i+1$ set:
\begin{eqnarray}\label{gauge}
h_j=h_{i+1}g_{i+1}g_{i+2}\dots g_{j-2}g_{j-1}
%h_1&=&h_ig_{i+1}\dots g_Ng_1, h_2=h_ig_{i+1}\dots g_Ng_1g_2, \dots, h_{i-1}=h_ig_{i+1}\dots g_Ng_1
%\dots g_{i-1}, \\ \nonumber h_{i+1}&=&h_ig_{i+1}, \dots h_N=h_ig_{i+1}\dots g_N,  
\end{eqnarray} 
Here here the product is cyclic. Denote such $N$-tuple $(h_1,\dots, h_N)$ by $h_g$.

The gauge transformation by the element $h_g$ brings $g=(g_1,\dots, g_N)$ to \\
$g^{h_g}=(1, \dots, h_{i+1}(g_{i+1}\dots g_Ng_1g_2\dots g_{i-1}g_i)h_{i+1}^{-1},\dots, 1)$. This identifies the $G^N$ gauge orbit through $(g_1,\dots, g_N)$ with the $G$ conjugation orbit through $g_1\dots g_N$.
Thus, we constructed the mapping 
\[
G^N/G^N\simeq G/G
\]
which is easy to prove to be an isomorphism. 

We can chose $h_{i+1}$ in such a way that 
$h_{i+1}(g_{i+1}\dots g_Ng_1g_2\dots g_{i-1}g_i)h_{i+1}^{-1}\in H$ is an element of the Cartan subgroup in $G$. This gives an isomorphism
\[
G^N/G^N\simeq G/G\simeq H/W
\]
for each $i=1,\dots, N$.

For each $i=1,\dots, N$, this gives an isomorphism of Hilbert spaces
\begin{equation}\label{funct-Cartan}
\phi_i: \CH_{\mu}\simeq \CH_\mu^H=L_2(H\to (V_{\mu_1}\otimes \dots\otimes V_{\mu_N})[0])^W
\end{equation}
Here the Weyl group $W$ acts as $w(f)(h)=wf(w^{-1}(h))$. We took into account that the Weyl group acts naturally on the zero weight subspace of any $G$-module. 

This isomorphism acts on functions as 
\[
\phi_i(f)(h)=\pi_\mu(h_g)f(g)=f(1,\dots, h,\dots, 1),
\]
where $h_g$ is as above, $\pi_\mu(h)$ is as in (\ref{eqprop}) and $h=h_i(g_{i+1}\dots g_Ng_1g_2\dots g_{i-1}g_i)h_i^{-1}$. 

For the scalar product on $\CH^H_\mu$ we have:
\[
(f,g)=\int_H (f(h),g(h)) |\delta(h)|^2 dh
\]
Here $H\subset G$ is the Cartan subgroup and $\delta(h)$ is the denominator 
in the Weyl character formula.

\subsubsection{Quantum Hamiltonians} 
I the theorem below we summarized computations of the action of
quantum Hamiltonians on the space $\CH_\mu^H$ with the gauge fixing map $\phi_N: \CH_\mu\to \CH_\mu^H$. We show that
operators $D_k=H^{(k+1)}-H^{(k)}$ are topological Knizhnik-Zamilodchikov-Bernard operators
and that $H^{(N)}$ is the spin Calogero-Moser Hamiltonian.

\begin{theorem} Operators $H^{(N)}$ and $D_j=H^{(j+1)}-H^{(j)}$ for $j=1,\dots, N-1$ act as
\[
H^{(N)}_2=\Delta +2\sum_{\alpha>0}\frac{\pi^V(f_\alpha e_\alpha)}{(1-h_\alpha)(1-h_{-\alpha})}-||\rho||^2
\]
and
\[
D_j=(h^{(j)}, \frac{\pa}{\pa\lambda})-\sum_{k=1}^{j-1}r_{kj}(\lambda)+\sum_{k=j+1}^n r_{jk}(\lambda)
\]
where $\Delta=(\frac{\pa}{\pa \lambda}, \frac{\pa}{\pa \lambda})$ is the Laplacian on $\hfr\subset \g$ determined by the Killing form, $h_\alpha=e^{(\alpha,\lambda)}$ and $r(\lambda)$ is Felder's dynamical $r$-matrix \cite{Felder}:
\[
r(\lambda)=-\frac{1}{2}\sum_{i=1}^r h_i\otimes h_i-\sum_\alpha \frac{e_{-\alpha}\otimes e_\alpha}{1-h_{-\alpha}}
\]
\end{theorem}
Note that we consider compact simple $G$ which correspond to $h_\alpha=e^{iq_\alpha}$ with $q_\alpha\in \RR$.

As always explicit formulae for higher Hamiltonians $H_k^{(j)}$ are more complicated.

\subsection{The multi-time propagator}

\subsubsection{The propagator and the graphical calculus.} 

\begin{figure}[t!]\label{T-function}
\begin{center}
\includegraphics[width=1.0\textwidth, angle=0.0, scale=0.7]{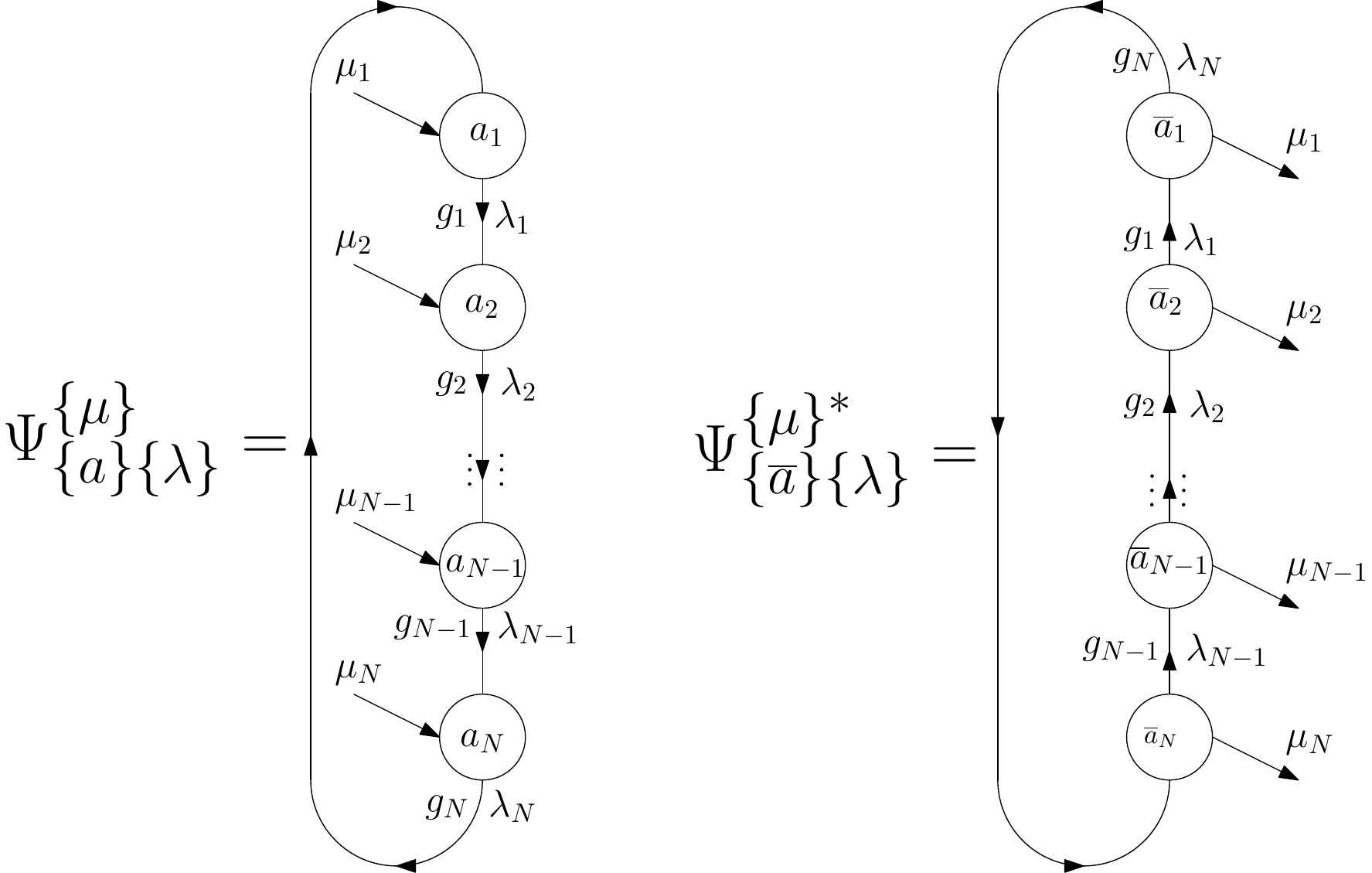}
\caption{\label{circle} The $N$-point trace function and the dual trace function. Holonomies $g_1,\dots, g_N$ are arguments of the trace function.}
\end{center}
\end{figure}

The multi-time propagator is simply an element of the semigroup generated by all commuting 
Hamiltonians:
\[
U^{\{\mu\}}_{\{A\}}=\exp(-\sum_{k,j} H_k^{(j)}A_{k,j})
\]
It has $Nr$ independent times $A_{k,j}$ and satisfies natural composition property
\[
U^{\{\mu\}}_{\{A\}}U^{\{\mu\}}_{\{B\}}=U^{\{\mu\}}_{\{A+B\}}
\]
where $(A+B)_{k,j}=A_{k,j}+B_{k,j}$.

The propagator is an  integral operator acting on the space $\CH_{\mu}$. In the usual way, from the spectral decomposition we derive its kernel:

\begin{equation}\label{prop}
U^{\{\mu\}}_{\{A\}}(\{g\}, \{g'\})=\sum_{\{\lambda\},\{a\}} \exp(-\sum_{i=1}^N c_2(\lambda_i)A_i)\Psi^{\{\mu\}}_{\{a\},\{\lambda\}}(g_1,\dots, g_N)
\Psi^{\{\mu\}}_{\{\overline{a}\},\{\lambda\}}(g_1',\dots, g_N')^*
\end{equation}

Here $\Psi^{\{\mu\}}_{\{\overline{a}\},\{\lambda\}}(g_1,\dots, g_N)^*\in V_{\mu_1}^*\otimes \dots \otimes V_{\mu_N}^*$ are dual trace functions
defined as 
\begin{align}\label{dula-t-fncn}
\Psi^{\{\mu\}}_{\{\overline{a}\},\{\lambda\}}(g_1,\dots, g_N)^*=Tr_{V_{\lambda_N}^*}(\pi_{\lambda_N}^*(g_N)b_1\dots (\pi_{\lambda_1}^*(g_1)\otimes id_{V_{\mu_1}})  (b_2\otimes id_{V_{\mu_2}} \otimes id_{V_{\mu_1}})\dots \nonumber \\ (b_{N-1}\otimes id_{V_{\mu_{N-1}}} \otimes\dots\otimes  id_{V_{\mu_1}})(\pi_{\lambda_1}^*(g_{N-1})id_{V_{\mu_1}}\otimes \dots\otimes id_{V_{\mu_{N-1}}}))
\end{align}
For $a_i: V_{\lambda_i}\to V_{\mu_i}\otimes V_{\lambda_{i-1}}$ the morphism $\overline{a}_i: V_{\lambda_i}^*\otimes V_{\mu_i}\to V_{\lambda_{i-1}}^*$ is a $\g$ linear map 
obtained by partial dualizing of $a^\dag:  V_{\mu_i}\otimes V_{\lambda_{i-1}}\to V_{\lambda_i}$. 

Note that this kernel pointwise is a linear operator in $End(V_{\mu_N}\otimes \dots \otimes V_{\mu_1})\simeq (V_{\mu_N}\otimes \dots \otimes V_{\mu_1})\otimes (V_{\mu_1}^*\otimes \dots \otimes V_{\mu_N}^*)$.

From now on we will focus on multi-time propagators with non-zero times
for quadratic Hamiltonians only. We will also be using graphical calculus to visualize algebraic 
operations. It is a version of Penrose's
graphical calculus. 
For trace functions $\Psi$ we will write a circle marked points corresponding to vertex operators $\{a_i\}$ andwith outgoing semi-open intervals colored by representations $V_{m_i}$.
For an edge colored by a highest weight $\lambda$ and an element $g\in G$ we assign 
$\pi^\lambda(g)$.
To a vertex colored by a $G$-invariant linear map $a:V_\lambda \to V_\mu\otimes V_\nu$ where
$\lambda$ is the color of the incoming edge and $\mu$ and $\nu$ are colors of outgoing edges, we assign the linear map $a$. Trace function is the contraction of these maps, which correspond to 
assembling (gluing)  the graph on Fig. \ref{prop-spec} from edges and vertices. 

\begin{figure}[t!]\label{prop-spec}
\begin{center}
\includegraphics[width=0.7\textwidth, angle=0.0, scale=0.7]{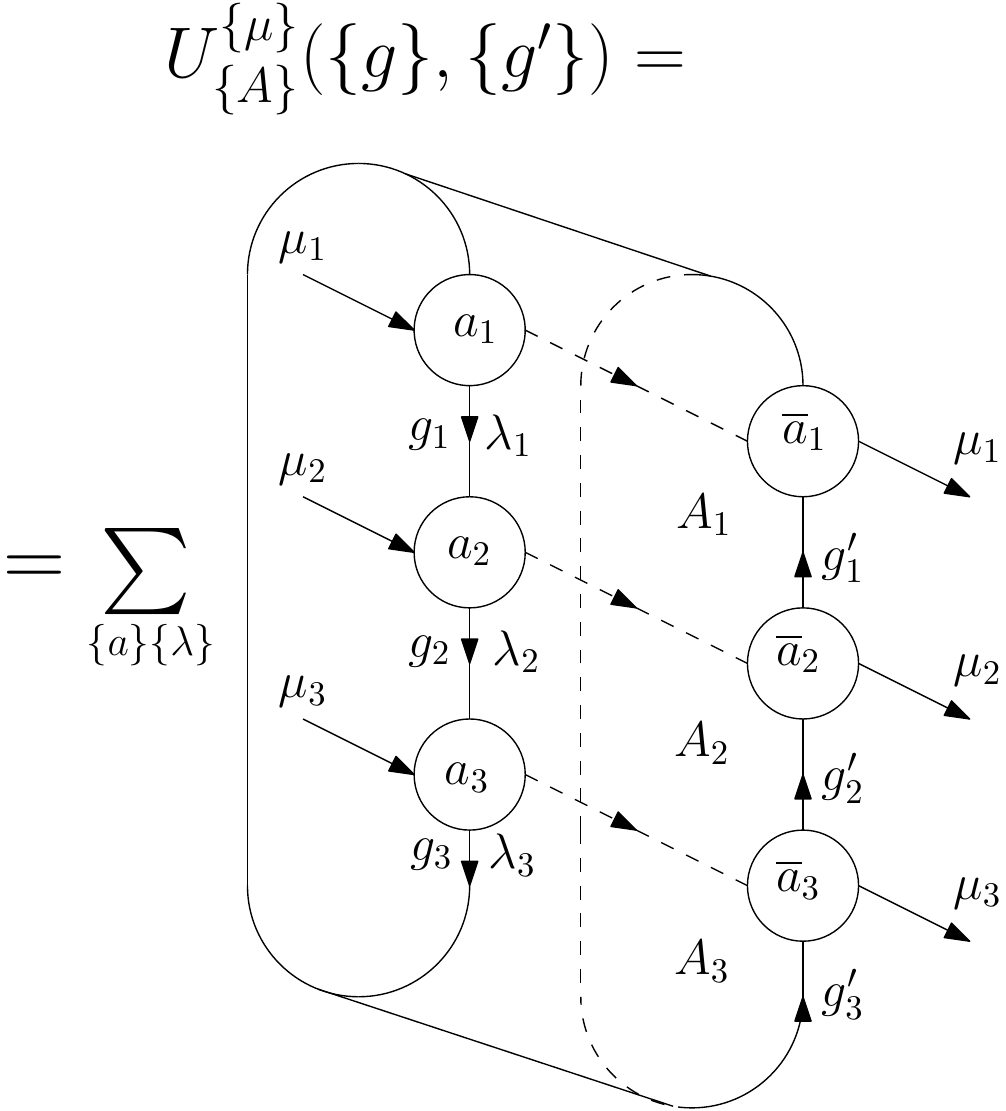}
\caption{Here is the graphical representation of the propagator (\ref{prop}) for $N=3$.
We put the "color" $\lambda_i$ on the $i$-th region. Both segments of the boundary 
of this region which also belong to the boundary of the cylinder inherit this color. Dashed lines 
are colored by representations $\mu_i$. Colors $\lambda_i$  they satisfy the Clebsch-Gordan rules
at each dashed line, as it is shown below. There $V_\lambda\subset V_\mu\otimes V_{\lambda'}$ Here we use
the orientation of the cylinder to distinguish regions that are left and right to the dashed line. }
\end{center}
\end{figure}

\begin{figure}[t!]\label{CG-rules}
\begin{center}
\includegraphics[width=0.8\textwidth, angle=0.0, scale=0.4]{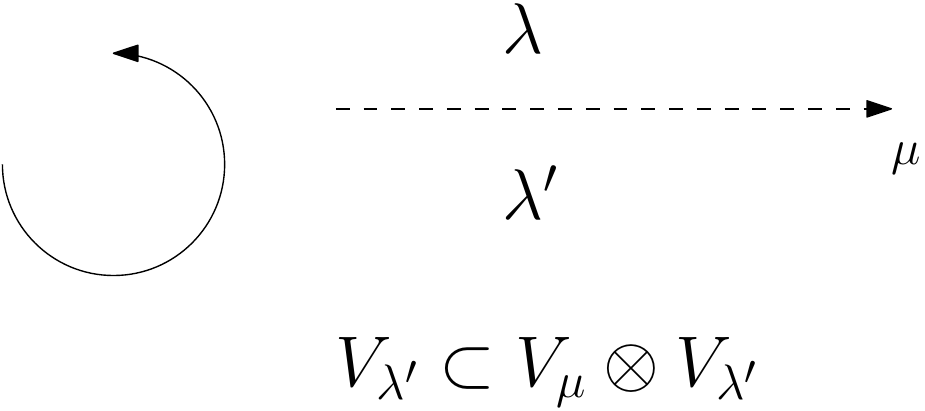}
\caption{Coloring of faces and edges are related by Clebsch-Gordan rules and uses the orientation of the surface. }
\end{center}
\end{figure}

\subsubsection{The integral formula}

\begin{theorem} The kernel (\ref{prop}) can be written as the following integral
\begin{eqnarray}\label{prop-int-cyl}
U^{\{\mu\}}_{\{A\}}(\{g\}, \{g'\})=\int_{G^N} Z_{A_1}({g_1'}^{-1}h_N^{-1}g_1h_1)Z_{A_2}({g_2'}^{-1}h_1^{-1}g_2h_2)
\dots \\ Z_{A_N}({g_{N-1}'}^{-1}h_{N-1}^{-1}g_Nh_N) \pi^{\mu_N}(h_N)\otimes \dots \otimes \pi^{\mu_1}(h_1) dh_1\dots dh_N
\end{eqnarray}
Here 
\[
Z_A(g)=\sum_\lambda \chi_\lambda(g)
dim(\lambda) \exp(-Ac_2(\lambda))
\]
is the partition function of the two-dimensional Yang-Mills theory with the 
boundary holonomy $g\in G$ \cite{Mig}\cite{Witt-1}.
\end{theorem}

To prove this statement one should use standard integral identities from Appendix \ref{A} This integral formula also has a natural graphical representation, see Fig. \ref{prop-int}.

\begin{figure}[t!]\label{prop-int}
\begin{center}
\includegraphics[width=1.0\textwidth, angle=0.0, scale=0.7]{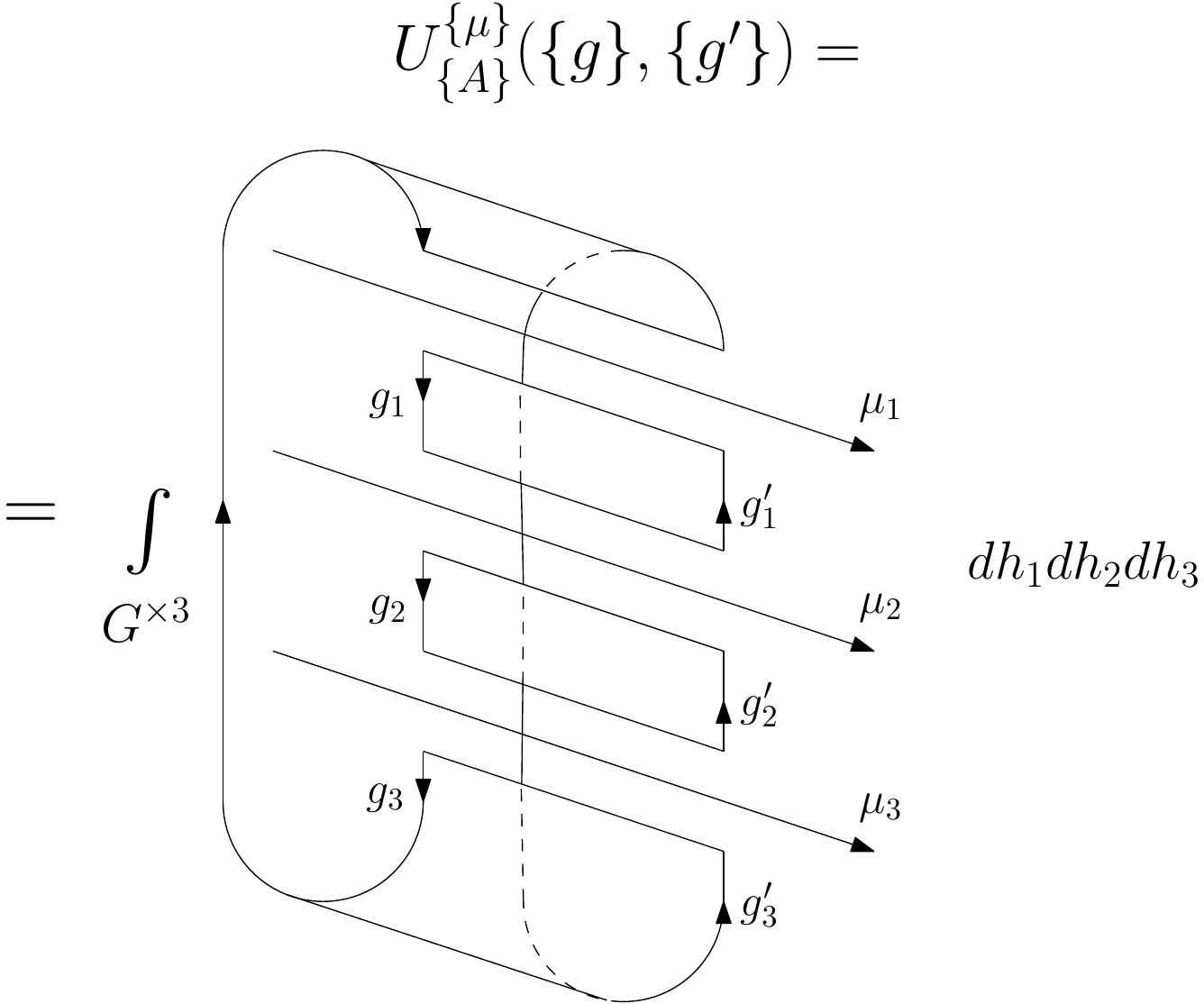}
\caption{Here is the graphical representation of the propagator (\ref{prop}) in 
the integral form (\ref{prop-int-cyl}) for $N=3$. 
Lines colored by representations $\mu_i$ should be regarded as open Wilson lines
and each region which is homeomorphic to a disc contributes as the corresponding factor
in (\ref{prop-int-cyl}). }
\end{center}
\end{figure}

\begin{figure}[t!]\label{disc}
\begin{center}
\includegraphics[width=0.6\textwidth, angle=0.0, scale=0.7]{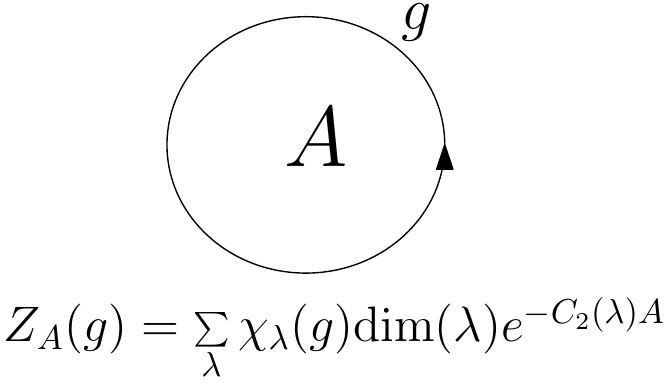}
\caption{The partition function for a disc. }
\end{center}
\end{figure}

The integral formula for the propagator reminds a lot the partition function of pure 2-dimensional Yang-Mills theory \cite{Mig}\cite{Witt-1}. An unusual feature of this formula are "open" Wilson lines which go 
to the outside of the cylinder as outer edges. We 
build on this analogy in the next section where will define 2-dimensional 
Yang-Mills theory of surfaces with corners and embedded open graphs.

\section{Quantum two dimensional Yang-Mills with corners}

In this section we define two dimensional Yang-Mills theory on space time manifolds with non-removable corners. We will not use the language of 2-categories, though it is natural in such setting.

\subsection{Surfaces with open graphs.}

Let $\Sigma$ be a topological oriented compact surface, possibly with a boundary $\pa \Sigma$.
Let $\Gamma$ be a graph, possibly with 1-valent vertices. 
We assume that $\Gamma$ is partially embedded in $\Sigma$ in the following sense: 
all vertices and all edges connecting vertices of valency greater than one are embedded to $\Sigma$\footnote{In the most general setting we can allow to have edges of $\Gamma$ which also belong to the
boundary of $\Sigma$.}. Some of the vertices possibly to $\pa \Gamma$, we call them boundary vertices. 
One valent vertices are connected only to boundary vertices. Both, one valent vertices and edges connecting them to boundary vertices do not belong to $\Sigma$, we will call them outer vertices and other edges.

Now let us remove outer vertices. We will be left with $\Gamma$ partially embedded
to $\Sigma$ with outer edges being open "away from $\Sigma$". 
We will call such pairs $(\Gamma,\Sigma)$ {\it open graphs on $\Sigma$}. An example of
an open graph on a surface is shown of Fig. \ref{graphs}

Note that each such graph defines a cell decomposition of $\Sigma$ with cells that are
not necessary contractible.

\begin{figure}[t!]\label{graphs}
\begin{center}
\includegraphics[width=0.5\textwidth, angle=0.0, scale=0.6]{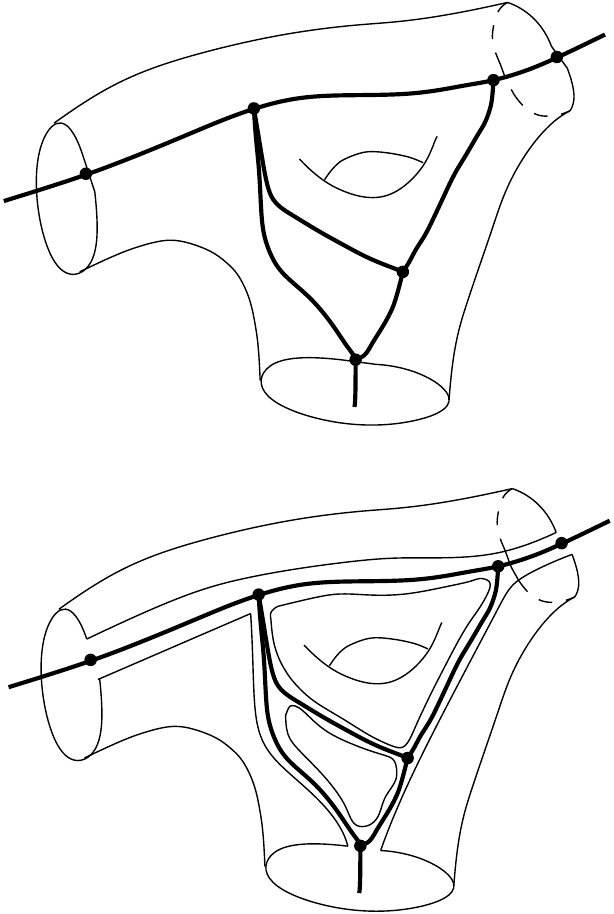}
\caption{An example of an open graph on a surface, upper figure. The same graph but enriched 
is shown on a lower figure. All edges of the graph are oriented and all 
connected components of the boundary of the surface are oriented
by the orientation of the surface ("counterclockwise") }
\end{center}
\end{figure}

Now for each open graph $\Gamma$ on $\Sigma$ we will construct {\it enriched graph} $\widehat{\Gamma}$ as follows:

\begin{itemize}

\item Add a circle to each connected component of $\Sigma\backslash \Gamma$ which follows
the boundary of this connected component.The orientation of this circle is induced by the orientation of 
$\Sigma$.

\item Replace intervals of $\pa \Sigma$ between boundary vertices of $\Gamma$ by segments of the corresponding inserted circle.

\end{itemize}
 
 \begin{figure}[t!]\label{gluing}
\begin{center}
\includegraphics[width=1.0\textwidth, angle=0.0, scale=0.6]{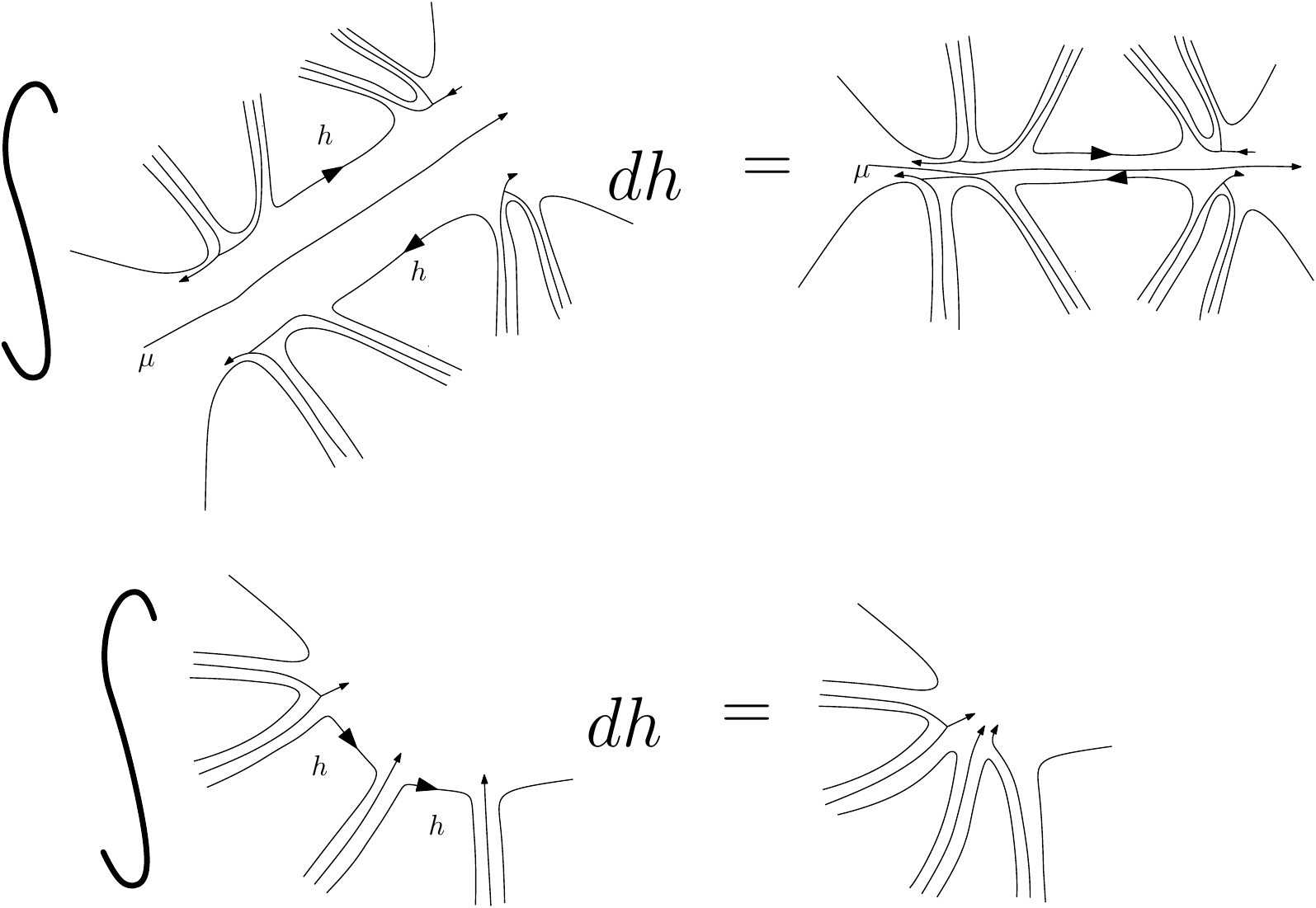}
\caption{An illustration to how partition surfaces are glued along intervals connecting 
non-removable corners. This may create multiple outer edges which can still be regarded as a singe 
outer edge colored by the tensor product of representations.}
\end{center}
\end{figure}

 Finally, let us define the gluing of surfaces with enriched graphs.
 
 A surface with an enriched graph in it can be glued
 to itself  by identifying two intervals at the boundary as it is shown on Fig. \ref{gluing}. 
 The upper figure corresponds to the case when the intervals do not shape a vertex.
 The lower figure corresponds to the case when the intervals share a vertex. In the first case
 they can belong to different connected components of the surface. 
 Such gluing of surfaces  correspond to compositions of partition functions
 which is also illustrated on Fig. \ref{gluing} and will be discussed in the next section.

\subsection{Quantum two-dimensional Yang-Mills theory on surfaces with open Wilson graphs.}

\begin{figure}[t!]\label{part-fncn}
\begin{center}
\includegraphics[width=0.6\textwidth, angle=0.0, scale=0.6]{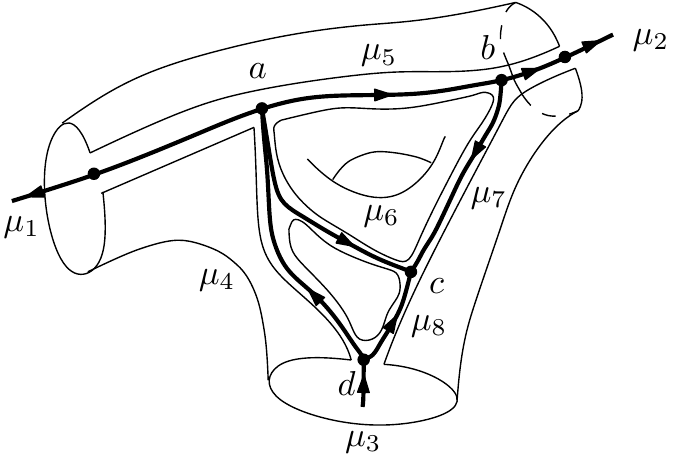}
\caption{Enriched colored open graph on a surface. Here $\mu_i$ are highest weights of representations assigned to edges and $a_, b, \dots$ are 
invariant vectors assigned to vertices. }
\end{center}
\end{figure}

First fix a simple compact Lie group $G$ and define a $G$-coloring of
a surface with an open graph as follows:
\begin{itemize}

\item We assign a finite dimensional representation of $G$ to each edge of the graph $\Gamma$. 

\item For each vertex of $\Gamma$ fix a total ordering of adjacent 
edges which agrees with the cyclic ordering induced by the orientation 
of $\Sigma$. Enumerate adjacent edges $e_1,\dots, e_m$ according to this total ordering. Let $V_{e_1}, \dots , V_{e_m}$ be colors of representations assigned to adjacent edges. To the
corresponding vertex we assign a $G$-invariant vector in the tensor product
\[
a_v\in (V_{e_1}^{\epsilon_1}\otimes \dots \otimes V_{e_m}^{\epsilon_m})^G
\]
Here $V^+=V$ and $V^-=V^*$ and $+$ is for edges oriented outward the
vertex and $-$ is for edges oriented inward. 
\end{itemize}

We will call colored graphs Wilson graphs by analogy with Wilson loops
\cite{Witt-2}\cite{PR}.

Now let us assign a variable to each connected component of $\Sigma\backslash \Gamma$. We assume that this variable is additive with respect to gluing. For example if a surface would be equipped with a volume
form, an area would be such additive variable. In the previous section
when the surfaces are cylinders, Euclidean times evolution for spin Calogero-Moser model Hamiltonians are such variables.

Now let us describe the space of boundary states. Let $\pa_\alpha\Sigma$ be
connected components of the boundary of $\Sigma$, $\pa \Sigma=\sqcup_\alpha \pa_\alpha \Sigma$. If a connected component
does not have boundary vertices of $\Gamma$, it is a circle. If it does,
it is a union of intervals to which boundary vertices of $\Gamma$ partition 
this connected component. We will call these intervals and circles {\it parts} of $\pa \Sigma$.

To  connected component $\pa_\alpha\Sigma$ we assign the following
spaces of boundary states:
\begin{itemize}

\item If $\pa_\alpha\Sigma$ is a circle we assign to it
\[
\CH(\pa_\alpha\Sigma)=L_2(G)^G
\]
It is a Hilbert space where characters of finite dimensional 
irreducible representations of $G$ form an orthonormal basis.

\item If the connected component $\pa_\alpha\Sigma$ has 
boundary vertices of $\Gamma$ on it, we assign the following space
to this boundary component:
\[
\CH(\pa_\alpha\Sigma)=L_2(G^{N_\alpha}, \otimes_{i=1}^{N_\alpha}(V^{\epsilon_{1,i}}_{e_{1,i}}\otimes \dots \otimes V^{\epsilon_{m_i,i}}_{e_{m_i,i}}))^{G^{N_\alpha}}
\]
Here $V_{e_{1,i}}^{\epsilon_{1,i}}\otimes \dots \otimes V_{e_{m_i,i}}^{\epsilon_{m_i,i}}$ is the tensor
product of colors of outer edges adjacent to boundary vertex $i$, $G^{N_\alpha}$ acts by gauge transformations on $G^{N_\alpha}$. On the tensor product it acts as 
\[
(g_1,\dots, g_{N_\alpha})(\otimes_{i=1}^{N_\alpha}(x_{e_{1,i}}\otimes \dots \otimes x_{e_{m_i,i}}))=\otimes_{i=1}^{N_\alpha}(g_ix_{e_{1,i}}\otimes 
\dots \otimes g_ix_{e_{m_i,i}})
\]
 
\end{itemize}

The partition function of the 2D Yang-Mills with for a surface with an open Wilson graphs is defined as a vector the following vector in the space of boundary states:
\begin{equation}\label{part-fncn}
Z_{\Gamma, A}(\{g\})=\int_{G^{E_{int}}} \prod_{D\in \Sigma\backslash \Gamma}
Z_{D,A_D}(\{g, h\}_{\pa D}) W_\Gamma(\{g, h\}) \prod_{e\in E_{int}} dh_e
\end{equation}
Here $D$ are parts of $\Sigma$ bounded by circles in the enriched 
open graph $\Gamma$, $hol_{\pa D}$ is the product of group elements $h_e$ assigned to internal edges in $\pa D$ and group elements $g_e$ corresponding to boundary edges of $\pa D$ in the counter clock-wise order starting from any edge. 
\begin{equation}\label{D-part-f}
Z_{D,A}(hol_{\pa D})=\sum_{\lambda}\chi_\lambda(hol_{\pa D}){\dim(\lambda)}^{1-2g}e^{-c_2(\lambda)A}
\end{equation}
where $g$ is the genus of the surface obtained from $D$ by gluing a disc
to $\pa D$, i.e. by "closing" the surface.

For example $hol_{\pa D}=1$ can be interpreted as if the boundary of $D$ does
exists, i.e. shrinks to a point. Thus  $Z_{D,A}(1)$ is the partition function of
2D Yang-Mills on the "closure" of $D$ \cite{Witt-2}. If $D$ is a disk
\begin{equation}\label{disk-part-f}
Z_{D,A}(g)=\sum_{\lambda}\chi_\lambda(g){\dim(\lambda)}e^{-c_2(\lambda)A}
\end{equation}
is the partition function of the 2D Yang-Mills on $D$ with boundary holonomy $g$. It is easy to see that for a triangulated surface $D$ the partition function (\ref{D-part-f}) can be composed from elementary partition functions (\ref{disk-part-f}) \cite{Witt-1}. This is the manifestation of locality of 2D Yang-Mills theory.

The function $W_\Gamma(\{g, h\})$ is the contraction of 
holonomies along internal edges evaluated in the representation corresponding to this edge and invariant vectors corresponding to 
vertices of $\Gamma$. Such function on $G^{E_{int}}$ is know as 
Wilson graph evaluated on a graph connection see for example \cite{Witt-2}(?)\cite{PR}. Wilson loop is a particular case of such a function. 

Partition functions (\ref{part-fncn}) satisfy the gluing property. 

\begin{theorem} Let $\Sigma_1$ and $\Sigma_2$ be two surfaces with boundary edges 
$e_1\in \pa \Sigma_1$ and $e_2\in \Sigma_2$ being identified as it is shown of Fig. \ref{gluing}.
Let $h$ be the holonomy along identified  edges then 
\[
\int_G Z(\Sigma_1;h)\otimes \otimes \pi_\mu(h)Z(\Sigma_2;h)dh=Z(\Sigma_1\cup_{e_1\simeq e_2}\Sigma_2)
\]
where in the partition function on the right there is an extra Wilson line colored by $V_\mu$
"sandwiching" between $\Sigma_1$ and $\Sigma_2$. And a similar gluing identity we have for the 
partition functions in case we glue two neighboring  intervals on $\pa \Sigma$.
\end{theorem}

The proof follows immediately from integral identities in the Appendix.

\subsection{Point observables}

Wilson graphs can be regarded as line observables in 2D Yang-Mills. They has been studies 
since many years ago, see for example \cite{Mig}\cite{Witt-1}\cite{Witt-2}\cite{CMR}\cite{PR}.

Let $D\subset \Sigma$ be a disc. Let us assign to $D$ the following function on $G/G$:

\[
\CO_{A}^F=\sum_\lambda \chi_\lambda (g) F_\lambda \dim(\lambda)e^{-Ac_2(\lambda)}
\]
When $F_\lambda=1$ for all $\lambda$ this is the partition function of 2D Yang-Mills on a disc.
We will call such function the partition function of a disc $D$ with {\it topological point observables}
inserted in the disc. The position of the point observable on $D$ is not given since it is a topological 
observable.  One should think of the function
$F$ as "sitting" at a point in the interior of $D$, confined to $D$ by the Wilson loop at the boundary.

Naturally, if we glue a surface to $D$ along a connected component $\pa_\alpha \Sigma$ of the 
boundary $\pa \Sigma$ of the surface $\Sigma$ the resulting partition function is the partition function with the point observable $F$ inserted in $\Sigma$. 

The insertion of such observable in the partition function for a surface with an open Wilson
graph $\Gamma$ on it depends on into which region of $\Sigma\backslash \Gamma$
it is inserted. But it does not depend at which point of the region it is inserted.  

If we put such point observable on a cylinder, the corresponding integral operators 
\[
U^F_A(g,h)=\sum_\lambda \chi_\lambda(g) \overline{\chi_\lambda(h)}F_\lambda \exp(-c_2(\lambda)A)
\]
form commutative algebra
\[
U_A^F\ast U_B^G=U^{FG}_{A+B}
\]
where $(FG)_\lambda=F_\lambda G_\lambda$. 

An example of such observable was introduced in \cite{Witt-2} as the 
insertion making the partition function independent of the orientation. 
If $\lambda^*=-w_0(\lambda)$\footnote{Here $w_0$ is the longest element of the Weyl group.} is the highest weight if the dual representation to $V_\lambda$, such observable corresponds to $F_\lambda=\delta_{\lambda, \lambda^*}$.

\section{Conclusion}

Here we will list some open problems and future directions.

\begin{enumerate}

\item The study of large $n$ limit in the $SU_n$ 2D YM theory,  is one of the 
most interesting directions for a number of well known reasons. When $\Sigma$ 
is a sphere there is a phase transition in this limit \cite{DK} at $A=\pi^2$. Similar behavior for 
the cylinder and for a sphere  with three holes was studied in \cite{GM}. For the relation 
between the 2D YM in large $n$ limit and the string theory see \cite{GT}.
It would be interesting to extend these results to the case of open
Wilson graphs.

\item In the follow up paper we will
describe the 2D YM theory with mixed gauge groups. Particular cases of such models are 
related to open spin Calogero-Moser spin chains, see \cite{SR} and for its classical counterpart \cite{R4}.

\item  There is a natural $q$-deformation of the $N$-spin CM systems described in \cite{dCRS}.
When $q$ is a root of unity this system is closely related to the quantum Chern-Simons theory \cite{W} for $U_q(\g)$
at roots of unity \cite{RT}. Corresponding deformations of 2D YM theory are invariants of twisted circle bundles over surfaces. Quadratic Casimirs are special in this setting 
because the ribbon structure in the corresponding modular tensor category is given by the exponents of quadratic Casimirs. 
Classical counterparts of these deformations are "relativistic" versions of 
spin Calogero-Moser models on moduli space of flat connections constructed in \cite{AR}.

\item As it was already mentioned partition functions of 2D YM theory can be regarded as a semiclassical limit of invariants of twisted circle bundles over surfaces corresponding to the quantum Chern-Simons theory with a compact Lie group. This makes the study of a quantum 2D YM theory for a non-compact Lie group even more interesting. It may be a toy model for
quantum complex Chern-Simons theory of a Chern-Simons theory for a non-compact real
Lie group. Quantum $N$-spin systems for noncompact real forms of simple complex Lie groups
were studied in \cite{SR}. 

\end{enumerate}

\appendix 

\section{On quantum moment maps.}\label{M}

\subsection{Quantum moment maps for $G$-actions}
The left and the right quantum moment maps act as:
\[
\widehat{\mu}_\ell(X)f(g)=\frac{d}{dt}f(e^{-tX}g)|_{t=0}, \ \ 
\widehat{\mu}_r(X)f(g)=\frac{d}{dt}f(ge^{tX})|_{t=0}
\]

For the left moment map acting on monomials in $U(\g)$ we have
\[
\widehat{\mu}_\ell(x)=\widehat{\mu}_\ell(X_1)\dots \widehat{\mu}_\ell(X_k)f(g)
=\frac{\pa^k}{\pa t_1\dots \pa t_k}f(e^{-t_kX_k}\dots e^{-t_1X_1}g)
\]
and a similar formula for the right action.

For $X\in \g\subset U(\g)$ the left and the right quantum moment maps are related as:
\begin{equation}\label{LR-relation}
\widehat{\mu}_r(X)f(g)=-\widehat{\mu}_\ell(Ad_g(X))f(g)
\end{equation}
This relation extends to monomials $x=X_1\dots X_k$ as
\[
\widehat{\mu}_r(x)=\widehat{\mu}_\ell(Ad_g(S(x))
\]

Indeed
\begin{eqnarray}\label{non-mon}
\widehat{\mu}_r(x)f(g)=\widehat{\mu}_r(X_1)\dots\widehat{\mu}_r(X_k)f(g)=
\frac{\pa^k}{\pa t_1 \dots \pa t_k}f(ge^{t_1X_1}\dots e^{t_kX_k})\\=(-1)^k\widehat{\mu}_\ell(Ad_g(X_k))\dots \widehat{\mu}_\ell(Ad_g(X_1))f(g)=\widehat{\mu}_\ell(Ad_g(S(x)))f(g)
\end{eqnarray}

Lie group $G$ acts on its Lie algebra $\g$ via adjoint action $g: X\mapsto Ad_g(X)$.
This action lifts to the action on $U(\g)$ which will also denote by $Ad_g$. The Lie group $G$ also acts naturally on differential
operators on $G$, $h: D\mapsto D^h$
\[
D^h f(g)={h_L^*}^{-1}(Df(hg))
\]
Here ${h_L^*}^{-1}f(g)=f(h^{-1}g)$.

Quantum moment maps are $G$-equivariant:
\[
\widehat{\mu}_\ell^h(x)=\widehat{\mu}_\ell(Ad_h(x))
\]
Let $\{e_i\}$ be a basis in $\g$ and $z=\sum_{i_1,\dots,i_k}f^{i_1,\dots, i_k}e_{i_1}\dots e_{i_k}$It is clear that differential operator $D=\widehat{\mu}_\ell(z)$ is $G$-invariant iff tensor $f$ is  $G$-invariant, i.e. iff $z\in Z(\g)$.

\begin{lemma}\label{z} Let $z\in Z(\g)$, then
\[
\widehat{\mu}_r(z)=\widehat{\mu}_\ell(S(z)
\]
\end{lemma}
\begin{proof}
Indeed,  we have:
\begin{eqnarray}
\widehat{\mu}_rf(g)=\sum_{\{i\}}f^{\{i\}}\widehat{\mu}_r(e_{i_1})\dots\widehat{\mu}_r(e_{i_k})f(g)=
\sum_{\{i\}}f^{\{i\}}\frac{\pa^k}{\pa t_1 \dots \pa t_k}f(ge^{t_1e_{i_1}}\dots e^{t_ke_{i_k}})\\=\sum_{\{i\}}f^{\{i\}}\widehat{\mu}_\ell(e_{i_k})\dots \widehat{\mu}_\ell(e_{i_1})f(g)=\widehat{\mu}_\ell(S(z))f(g)
\end{eqnarray}
\end{proof}

\subsection{Quantum moment maps for $G^N$ actions } For the left and right action of $G^N$ on itself:
\[
\widehat{\mu}_L(X_1,\dots, X_N)=\widehat{\mu}^{(1)}_\ell(X_1)+\dots +\widehat{\mu}^{(N)}_\ell(x_N), 
\ \ \widehat{\mu}_R(X_1,\dots, x_N)=\widehat{\mu}^{(1)}_r(X_1)+\dots +\widehat{\mu}^{(N)}_r(X_N)
\]
where $\widehat{\mu}^{(i)}_{\ell,r}$ are left and right moment maps for the $i$-th factor.

For the gauge action of $G^N$ on $G^N$ by gauge transformations we have:
\[
\widehat{\mu}(X_1,\dots, X_N)=\widehat{\mu}_L(X_1,X_2,\dots, X_{N})+\widehat{\mu}_R(X_2, X_3\dots, X_N, X_1)
\]
These Lie algebra homomorphisms extend uniquely to algebra homomorphisms
$\widehat{\mu},\widehat{\mu}_L,\widehat{\mu}_R: U(\g)^{\otimes N}\to \Diff(G^N)$

Lemma \ref{z} implies that
\[
\widehat{\mu}_R(z_1\otimes \dots \otimes z_N)=\widehat{\mu}_L(S(z_1)\otimes \dots \otimes S(z_N))
\]
Because of this the natural algebra homomorphism
\[
\widehat{\mu}_L\otimes \widehat{\mu}_R: U(\g)^{\otimes N}\otimes U(\g)^{\otimes N}\to \Diff(G^N)
\]
filters through the algebra homomorphism
\[
\widehat{\mu}_L\otimes \widehat{\mu}_R: U(\g)^{\otimes N}\widetilde{\otimes}_{Z(\g)^{\otimes N}} U(\g)^{\otimes N}\to \Diff(G^N)
\]
Here  $U(\g)^{\otimes N}\widetilde{\otimes}_{Z(\g)^{\otimes N}} U(\g)^{\otimes N}$ is the tensor product over $Z(\g)^{\otimes N}$ where $(a_1z_1\otimes \dots \otimes a_Nz_N) \otimes_{Z(\g)^{\otimes N}} (b_1\otimes \dots \otimes b_N)=
(a_1\otimes \dots \otimes a_N) \otimes_{Z(\g)^{\otimes N}}(S(z_1)b_1\otimes \dots \otimes S(z_N)b_N)$.

The mapping $\phi_i$ induces a linear isomorphism $\End(\CH_\mu)\to \End(\CH^H_\mu)$ for
which we will use the same name $\phi_i$. 
Let us first compute how $\widehat{\mu}_L(U(\g))$ and $\widehat{\mu}_R(U(\g))$
act on $\CH^H_\mu$ and then the action of Hamiltonians $H^{(j)}_{d_k}$. 

\subsection{Quantum moment maps after gauge fixing}
Let $X\in \g\subset U(\g)$ and 
$X_j=1^{\otimes (j-1)}\otimes X \otimes 1^{\otimes (N-j)}$.
The following is the results of an elementary computation:
\[
\widehat{\mu_R}(X_j)f(g)=-\widehat{\mu_L}(Ad_{g_j}X_j)f(g), \ \   .
\]
\[
\widehat{\mu_R}(X_{j+1})+\widehat{\mu_L}(X_j)=\pi_j^{\mu_j}(X)f(g)
\]
for functions in $\CR_\mu$.
Applying gauge fixing at $i$ to these identities we obtain the left action of $X^{(j)}$ on $\CH^H_\mu$:
\begin{equation}\label{E1}
\phi_i(-\widehat{\mu_L}(X_{j+1})+\widehat{\mu_L}(X_{j}))f(h)=\pi_j^{\mu_j}(X) f(h), \ \ j\neq i
\end{equation}
\begin{equation}\label{E2}
\phi_i(-\widehat{\mu_L}(Ad_h(X_i))+\widehat{\mu_L}(X_{i-1}))f(h)=\pi_{i-1}^{\mu_{i-1}}(X) f(h), 
\end{equation}
where $f\in \CH^H_\mu)$.

Choose a basis $e_\alpha$ in each root subspace $\g_\alpha\subset \g$. 
From the equations above we can express the action of  
$e_{\alpha,j}$ in terms of $\pi_j^{\mu_j}(X)$and $h\in H$:
\[
\phi_i(\widehat{\mu_L})(e_{\alpha,j})=\frac{1}{1-h_\alpha}S^{j}_\alpha, 
\]
where 
\begin{equation}\label{E-i}
S^{i}_\alpha=\sum_{j=1}^N \pi_j^{\mu_j}(e_\alpha)
\end{equation}
\begin{equation}\label{E-j}
S^{j}_\alpha=h_\alpha(\pi_{i}^{\mu_{i}}(e_\alpha)+\dots+\pi_{j-1}^{\mu_{j-1}}(e_\alpha))+
\pi_{j}^{\mu_{j}}(e_\alpha)+\dots+\pi_{i-1}^{\mu_{i-1}}(e_\alpha), \ \ j\neq i
\end{equation}
Here $h_\alpha$ is the coordinate function on $H$ corresponding to the root $\alpha$,
$he_\alpha h^{-1}=h_\alpha e_\alpha$ and the summation is cyclic. 

For the left action of the Cartan subalgebra we have
\begin{equation}\label{l-act-cart}
\phi_i(\widehat{\mu}_L(H_{\alpha,j})-\widehat{\mu}_L(H_{\alpha,j+1}))f(g)=\pi_j(h)f(g)
\end{equation}
Here $j+1$ is defined cyclically $\mod N$.

After gauge fixing we have
\[
\phi_i(\widehat{\mu}_L(H_{\alpha,i})f(1,\dots, h,\dots, 1)=-i\frac{\pa}{\pa q_\alpha}f(1,\dots, h,\dots, 1)
\]
The action of $\phi_i(\widehat{\mu}_L(H_{\alpha,j})$ with $j\neq i$ can be computed from this formula
and from equations (\ref{E1})(\ref{E2}):
\[
\phi_i(\widehat{\mu}_L(H_{\alpha,j})=-\pi_{i}^{\mu_i}(H_\alpha)- \pi_{i+1}^{\mu_{i+1}}(H_\alpha)-\dots -\pi_{j}^{\mu_j}(H_\alpha)-i\frac{\pa}{\pa q_\alpha}
\]
Here the sum is cyclic.

\section{Quadratic Hamiltonian for $N=1$}\label{H1}

For$N=1$ the gauge invariance of functions 
from the space $\R_\mu$ is $\psi(hgh^{-1})=\pi(h)\psi(g)$.
The transformation property of functions from $\R_\mu$ imply:
\[
(\widehat{\mu}_\ell(X)+\widehat{\mu}_r(X))f(g)=\pi(X)f(g)
\]
Combining this with (\ref{LR-relation}) , for functions $f$ from $\R_\mu$ we have:
\[
(\widehat{\mu}_\ell(X)-\widehat{\mu}_r(Ad_g(X)))f(g)=\pi(g)f(g)
\]

Now let us compute $\widehat{\mu}_\ell(Y)\widehat{\mu}_\ell(X)f(g)$. From the definition (\ref{non-mon}) we have:
\[
\widehat{\mu}_\ell(Y)(\widehat{\mu}_\ell(X)-\widehat{\mu}_\ell(Ad_g(X)))f(g)=\pi(X)\widehat{\mu}_\ell(Y)f(g)
\]
similarly
\[
-\widehat{\mu}_\ell(Ad_g(Y))(\widehat{\mu}_\ell(X)-\widehat{\mu}_\ell(Ad_g(X)))f(g)=-\pi(X)\widehat{\mu}_\ell(Ad_g(Y)f(g)
\]
Combining these identities we have:
\begin{theorem} 
\begin{eqnarray}
\pi(X)\pi(Y)f(g)=(\widehat{\mu}_\ell(Y)-\widehat{\mu}_\ell(\widetilde{Z}))(\widehat{\mu}_\ell(X)-\widehat{\mu}_\ell(Z))\psi|_{Z=Ad_g(X), \widetilde{Z}=Ad_g(Y)}+\\
\widehat{\mu}_\ell([Y, Ad_g(X)])\psi-\widehat{\mu}_\ell(Ad_g([Y,X]))f(g)
\end{eqnarray}
\end{theorem}

 From here, setting $X=e_\alpha, Y=e_{-\alpha}$ and specializing $g$ to $h\in H$ we obtain:
 \[
 \pi(e_\alpha)\pi(e_{-\alpha})f(h)=(1-h_\alpha)(1-h_{-\alpha})\widehat{\mu}_\ell(e_{-\alpha})\widehat{\mu}_\ell(e_\alpha)f(h)+(1-h_\alpha) \widehat{\mu}_\ell(H_\alpha)f(h)
 \]
or
\[
\widehat{\mu}_\ell(e_{-\alpha})\widehat{\mu}_\ell(e_\alpha)f(h)=\frac{\pi(e_\alpha)\pi(e_{-\alpha})}{(1-h_\alpha)(1-h_{-\alpha})}f(h)+\frac{\widehat{\mu}_\ell(H_\alpha)}{(1-h_\alpha)}f(h)
\]
Setting $h=e^\lambda$ we can write $\widehat{\mu}_\ell(H_\alpha)=(\alpha, \frac{\pa}{\pa \lambda})$ and for the second Casimir we obtain
\[
\widehat{\mu}_\ell(c_2)f(h)=(\Delta +2\sum_{\alpha>0}\frac{\pi(e_\alpha)\pi(e_{-\alpha})}{(1-h_\alpha)(1-h_{-\alpha})}-D)f(h)
\]
where $D=\sum_{\alpha>0} \frac{1+h_\alpha}{1-h_\alpha}(\alpha, \frac{\pa}{\pa \lambda})$.

\section{Quantum quadratic Hamiltonians for $N>1$.}\label{H}

Here we will compute quadratic Hamiltonians from the action of $U(\g)^{\otimes N}$ on trace functions by left and right quantum moment map. Since trace functions form a basis in the space $\R_\mu$ this is equivalent to computing quantum Hamiltonians directly from the action on functions from $\R_\mu$ as it was done above for $N=1$, but it is easier.
 
\subsection{Dynamical $r$-matrices} \label{dyn-r} For $\lambda\in \hfr$ define $\xi_\alpha=e^{\frac{(\lambda, \alpha)}{2}}$,
$h_\alpha=\xi_\alpha^2$.
The dynamical $r$-matrix \cite{Felder} is the following $\g\otimes \g$-valued function on $\hfr$
\[
r(\lambda)=-\frac{1}{2}\sum_{i=1}^r h_i\otimes h_i-\sum_\alpha \frac{e_{-\alpha}\otimes e_\alpha}{1-h_{-\alpha}}
\]
or
\[
r(\lambda)=-\frac{1}{2}\sum_{i=1}^r h_i\otimes h_i+\sum_{\alpha>0}( \frac{e_{\alpha}\otimes f_\alpha}{h_{\alpha}-1}-\frac{h_\alpha f_{\alpha}\otimes e_\alpha}{h_{\alpha}-1}
\]

Here are some basic properties of $r(\lambda)$:

\begin{itemize}

\item $r(-\lambda)=r^{21}(\lambda)$,
\item $r(\lambda)+r(-\lambda)=-\Omega$,
\item $(e^{-\lambda})_2r_{12}(\lambda)(e^\lambda)_2=-\sum_{i=1}^r h_i\otimes h_i-r_{21}(\lambda)$,
\item $D_1D_2r(\lambda)D_1^{-1}D_2^{-1}=r(\lambda)$
\item The dynamical Yang-Baxter equation.

\end{itemize}

\subsection{Quantum topological Knizhnik-Zamolodchikov-Bernard equation}

All the computations from this sections can be done from the definition of
the space $\R_\mu$. However, it is more convenient to do them using the basis
of trace functions after the gauge fixing for $i=N$.

Choose a basis $e_\alpha$ is root subspaces of $\g$ and orthonormal 
(with respect to the Killing form) basis in $\hfr$. Let $c_2$ be the quadratic Casimir for $\g$:
\[
c_2=\sum_{i=1}^r h_i^2+\sum_{\alpha>0}(e_\alpha f_\alpha+f_\alpha e_\alpha)
\]
where $f_\alpha=e_{-\alpha}$. If $h_{\alpha_i}$ is a basis in $\hfr$ corresponding to simple roots,
$[e_\alpha, f_\alpha]=H_\alpha$, the Cartan part of the Casimir element can be written as
\[
\sum_{i=1}^r h_i^2=\sum_{i,j=1}^r H_{\alpha_i}(B^{-1})_{ij}H_{\alpha_j}
\]
where $B$ is the symmetrized Cartan matrix.

Denote by $\Omega$ the "mixed Casimir":
\[
\Omega=\sum_{i=1}^r h_i\otimes h_i+\sum_{\alpha>0}(e_\alpha \otimes f_\alpha+f_\alpha\otimes e_\alpha)
\]
If $\Delta$ is the compultiplication for $U(\g)$, we have
\[
\Delta c_2=c_2\otimes 1 +1\otimes c_2+2\Omega
\]

For $b\in Hom_\g(M_\nu, M_\mu\otimes V_\eta)$ where $M_\mu, M_\nu$ are Verma modules and $V_\lambda$ is an irreducible finite dimensional module with the highest weight $\eta$. We have
\[
bc_2(\nu)=(c_2(\mu)+c_2(\eta)+2\Omega)b
\]
This identity can be written in terms of dynamical $r$-matrices as
\begin{equation}\label{intertw-id}
\frac{c_2(\nu)-c_2(\mu)-c_2(\eta)}{2}b=(r_{12}(\lambda)+r_{21}(\lambda))b
\end{equation}
Here $\lambda\in \hfr$.

For $b_i\in \textup{Hom}_G(V_{\nu_i},V_{\nu_{i-1}}\otimes V_{\mu_i})$, with
$0=N$, consider the  trace function (\ref{tr-fncn}) evaluated at $g=(1,\dots, 1, e^\lambda)$,
i.e. for the gauge fixing with $i=N$:
\begin{equation}\label{N-t-fncn}
\Psi^{(\nu)}_{b,\mu}(\lambda)=Tr_a (b_1\dots b_N (e^\lambda)_a)
\end{equation}
Here we denote representation $V_{\mu_N}$ by index $a$ (auxiliary space)
and take trace over this space.

Denote
\[
d(\lambda)=\frac{1}{2}\sum_{\alpha>0}\frac{\xi_\alpha+\xi_{-\alpha}}{\xi_\alpha-\xi_{-\alpha}}H_\alpha
\]
\begin{theorem}
Trace function $F$ satisfies differential equations
\begin{equation}\label{kzb}
((h^{(i)}, \frac{\pa}{\pa\lambda})-\sum_{k=1}^{i-1}r_{ki}(\lambda)+\sum_{k=i+1}^n r_{ik}(\lambda)-d(\lambda)_i)\Psi^{(\nu)}_{b,\mu}(\lambda)=\frac{c(\nu_{i+1})-c(\nu_i)}{2}\Psi^{(\nu)}_{b,\mu}(\lambda)
\end{equation}
\end{theorem}

\begin{proof}
Let us apply the identity (\ref{intertw-id}) to the i-th factor in the definition of the trace
function. We have:
\[
\frac{c_2(\nu_i)+c_2(\mu_i)-c_2(\nu_{i+1})}{2}\Psi^{(\nu)}_{b,\mu}(\lambda)=Tr_a(b_1\dots b_{i-1}(r_{ia}(\lambda)+r_{ai}(\lambda))b_i\dots b_N (e^\lambda)_a)
\]

From the intertwining property of $b_i$ we have:
\[
b_1\dots b_{i-1}r_{ai}(\lambda)=(r_{ai}(\lambda)+r_{1i}(\lambda) + \dots +r_{i-1,i}(\lambda)b_1\dots b_{i-1}
\]

Now let us use the cyclicity of the trace:
\begin{eqnarray}
Tr_a(r_{ai}(\lambda)b_1\dots b_N (e^\lambda)_a)=Tr_a (b_1\dots b_N (e^\lambda)_ar_{ai}(\lambda)(e^{-\lambda})_a(e^\lambda)_a)= \\ \nonumber Tr_a (b_1\dots b_N (-\sum_kh^{(i)}_kh^{(a)}_k- r_{ia}(\lambda))(e^\lambda)_a)
=\\    \nonumber
-(h^{(i)}, \frac{\pa}{\pa \lambda}) Tr_a (b_1\dots b_N(e^\lambda)_a)-Tr_a (b_1\dots b_N r_{ia}(\lambda)(e^\lambda)_a))
\end{eqnarray}
Here we used properties of the dynamical $r$-matrix listed in section \ref{dyn-r}.

Now let us use the intertwining properties of $b_j$'s:
\[
b_{i+1}\dots b_N r_{ia}(\lambda)=(r_{i,i+1}(\lambda)+\dots + r_{iN}(\lambda)+r_{ia}(\lambda))b_{i+1}\dots b_N
\]
and
\[
b_ir_{ia}(\lambda)=(r_{ia}(\lambda)+m(r(\lambda))_i)b_i
\]
Here $m(a\otimes b)=ab$ for $a,b\in U(\g)$. 
\[
m(r(\lambda))=-\frac{1}{2} \sum_i h_i^2-\sum_\alpha \frac{e_{-\alpha}e_\alpha}{1-h_{-\alpha}}=-\frac{1}{2}c_2+d(\lambda)
\]
where $d(\lambda)$ is as above.

Thus,
\[
b_ib_{i+1}\dots b_N r_{ia}(\lambda)=(r_{ia}(\lambda)-\frac{c_2(\mu_i)}{2}+d(\lambda)_i+r_{i,i+1}(\lambda)+\dots + r_{iN}(\lambda))b_i\dots b_N
\]
Combining all identities we obtain (\ref{kzb}).

\end{proof}

\subsection{The quantum spin Calogero-Moser Hamiltonian}

For $b\in Hom_\g(M_\nu, M_\nu\otimes V)$, where $V$ is a finite dimensional representation of $\g$ and $M_\nu$ is the Verma module with the highest weight $\nu$, we define the trace function
\[
\Phi_{b,V}^{(\nu)}(\lambda)=Tr_a(b(e^\lambda)_a)
\]
Here the trace is takes over the space $M_\nu$ which we denote by index $a$ (auxiliary). For our purposes it is enough to consider this trace function as a Laurent power series in $e^\lambda$. 

Let $\Delta$ be the Laplacian on $\hfr$:
\[
\Delta=(\frac{\pa}{\pa \lambda}, \frac{\pa}{\pa \lambda})
\]
where $(.,.)$ is the bilinear form on $\hfr^*$ which is dual to the Killing from on $\hfr$.
We have:
\[
\sum_i h^2_i e^\lambda=\Delta e^\lambda
\]
We halso have the identity
\[
c_2(\nu) \Psi_{b,V}^{(\nu)}(\lambda)=Tr_a(b(c_2)_a(e^\lambda)_a)
\]

\begin{lemma} Following identities hold:
\begin{equation}\label{1}
Tr_{M_\nu}(be_\alpha e^\lambda)=\frac{\pi^V(e_\alpha)}{1-h_{\alpha}}\Psi_{b,V}^{(\nu)}(\lambda)
\end{equation}

\begin{equation}\label{2}
Tr_{M_\nu}(bf_\alpha e^\lambda)=\frac{\pi^V(f_\alpha)}{1-h_{-\alpha}}\Psi_{b,V}^{(\nu)}(\lambda)
\end{equation}
\end{lemma}
Indeed, the first identity follows from the intertwining property of $b$ and from the cyclicity of the trace:
\[
Tr(be_\alpha e^\lambda)=\pi^V(e_\alpha)\Psi_{b,V}^{(\nu)}(\lambda)+Tr(be^\lambda e_\alpha)=
\pi^V(e_\alpha)\Psi_{b,V}^{(\nu)}(\lambda)+h_\alpha Tr(b e_\alpha e^\lambda)
\]
The proof of (\ref{2}) is similar.

\begin{proposition}
\begin{equation}\label{a}
Tr_{M_\nu}(be_\alpha f_\alpha e^\lambda)=\frac{\pi^V(e_\alpha f_\alpha)}{(1-h_\alpha)(1-h_{-\alpha})}\Psi_{b,V}^{(\nu)}(\lambda)-\frac{h_\alpha}{1-h_\alpha}(\alpha, \frac{\pa}{\pa \lambda} \Psi_{b,V}^{(\nu)}(\lambda)
\end{equation}
\begin{equation}\label{b}
Tr_{M_\nu}(bf_\alpha e_\alpha e^\lambda)=\frac{\pi^V(f_\alpha e_\alpha)}{(1-h_\alpha)(1-h_{-\alpha})}\Psi_{b,V}^{(\nu)}(\lambda)-\frac{1}{1-h_\alpha}(\alpha, \frac{\pa}{\pa \lambda} \Psi_{b,V}^{(\nu)}(\lambda)
\end{equation}
\end{proposition}

\begin{proof}

\begin{eqnarray}
Tr_{M_\nu}(be_\alpha f_\alpha e^\lambda)=\pi^V(e_\alpha)Tr_{M_\nu}(b f_\alpha e^\lambda)+Tr_{M_\nu}(b f_\alpha e^\lambda e_\alpha)=\frac{\pi^V(f_\alpha)}{1-h_{-\alpha}}\Psi_{b,V}^{(\nu)}(\lambda)+\\
\nonumber
h_\alpha Tr_{M_\nu}(bf_\alpha e_\alpha e^\lambda)=
\frac{\pi^V(f_\alpha)}{1-h_{-\alpha}}\Psi_{b,V}^{(\mu)}(\lambda)+h_\alpha Tr_{M_\nu}(be_\alpha f_\alpha e^\lambda)-h_\alpha (\alpha, \frac{\pa}{\pa \lambda}\Psi_{b,V}^{(\mu)}(\lambda)
\end{eqnarray}

The proof of (\ref{b}) is similar.

\end{proof}

\begin{corollary}
\[
Tr_{M_\nu}(b(e_\alpha f_\alpha +f_\alpha e_\alpha) e^\lambda)=\frac{2\pi^V(f_\alpha e_\alpha)}{(1-h_\alpha)(1-h_{-\alpha})}\Psi_{b,V}^{(\nu)}(\lambda)-\frac{1+h_\alpha}{1-h_\alpha}(\alpha, \frac{\pa}{\pa \lambda} \Psi_{b,V}^{(\nu)}(\lambda)
\]
\end{corollary}
Here we use that $\Psi_{b,V}^{(\mu)}(\lambda)
$ is an element of the zero weight subspace in $V$.

Thus, we proved the following theorem.
\begin{theorem}
\begin{equation}\label{eigen}
\Delta \Psi_{b,V}^{(\nu)}(\lambda)
+2\sum_{\alpha>0}\frac{\pi^V(f_\alpha e_\alpha)}{(1-h_\alpha)(1-h_{-\alpha})}\Psi_{b,V}^{(\nu)}(\lambda)-
\frac{1+h_\alpha}{1-h_\alpha}(\alpha, \frac{\pa}{\pa \lambda}) \Psi_{b,V}^{(\nu)}=c_2(\mu)\Psi_{b,V}^{(\mu)}
\end{equation}
\end{theorem}

It is easy to see that if we replace Verma module $M_\nu$ by an irreducible 
representation $V_\nu$ with the highest weight $\mu$, the theorem and its proof hold. Note that when $V=V_{\mu_1}\otimes \dots \otimes V_{\mu_N}$  with the diagonal  action of $\g$, and $b=b_1\dots b_N$, where $b_i$ are the same intertwines as in (\ref{N-t-fncn}), the trace function $\Phi_{b,V}^{(\nu)}$ becomes the trace function $\Psi_{a,\mu}^{(\nu)}$ (\ref{N-t-fncn}) with $\nu=\nu_N$.

This proves that trace functions (\ref{N-t-fncn}) satisfy (\ref{eigen}).

\subsection{Normalization}

\begin{lemma}
\begin{equation}\label{delta-id}
\Delta(\delta^{-1})=-(\rho,\rho)-D(\delta^{-1}
\end{equation}
where $D=\sum_{\alpha>0}\frac{\xi_\alpha+\xi_{-\alpha}}{\xi_\alpha-\xi_{-\alpha}}(\alpha, \frac{\pa}{\pa \lambda})$
and $\delta$ is the denominator in the Weyl character formula:
\[
\delta=\prod_{\alpha>0}(\xi_\alpha-\xi_{-\alpha})
\]
\end{lemma}
\begin{proof}
Consider the trace function for the trivial representation $V$
with $b=id: M_\nu\to M_nu$. In this case $\Psi_{b,V}^{(\nu)}=\frac{e^{(\nu+\rho, \lambda)}}{\delta}$ is the character of $M_\nu$. In this case the identity (\ref{eigen})
immediately implies (\ref{delta-id}).
\end{proof}

\begin{proposition}
\[
\delta\circ (\Delta+D)\circ \delta^{-1}=\Delta-||\rho||^2
\]
\[
\delta\circ (h, \frac{\pa}{\pa \lambda})\circ \delta^{-1}=-\sum_{\alpha>0}\frac{\xi_\alpha+\xi_{-\alpha}}{\xi_\alpha-\xi_{-\alpha}}H_\alpha
\]
\end{proposition}

Define the normalized trace function as
\[
F_{b,V}^{(\nu)}(\lambda)=\delta^{-1}  \Psi_{b,V}^{(\nu)}
\]
We proved:
\begin{theorem} Normalized trace functions satisfy differential equations
\[
H^{(N)}_2F_{b,V}^{(\nu)}(\lambda)=c_2(\nu)F_{b,V}^{(\nu_N)}(\lambda), \ \ D_iF_{b,V}^{(\nu)}(\lambda)=\frac{c_2(\nu_{i+1})-c_2(\nu_i)}{2}F_{b,V}^{(\nu)}(\lambda)
\]
where $i=1,\dots, N$ and
\[
H^{(N)}_2=\Delta +2\sum_{\alpha>0}\frac{\pi^V(f_\alpha e_\alpha)}{(1-h_\alpha)(1-h_{-\alpha})}-||\rho||^2
\]
and
\[
D_i=(h^{(i)}, \frac{\pa}{\pa\lambda})-\sum_{k=1}^{i-1}r_{ki}(\lambda)+\sum_{k=i+1}^n r_{ik}(\lambda)
\]
\end{theorem}

\section{Useful identities with the Haar measure}\label{A}

\begin{figure}[t!]\label{int-2}
\begin{center}
\includegraphics[width=0.8\textwidth, angle=0.0, scale=0.6]{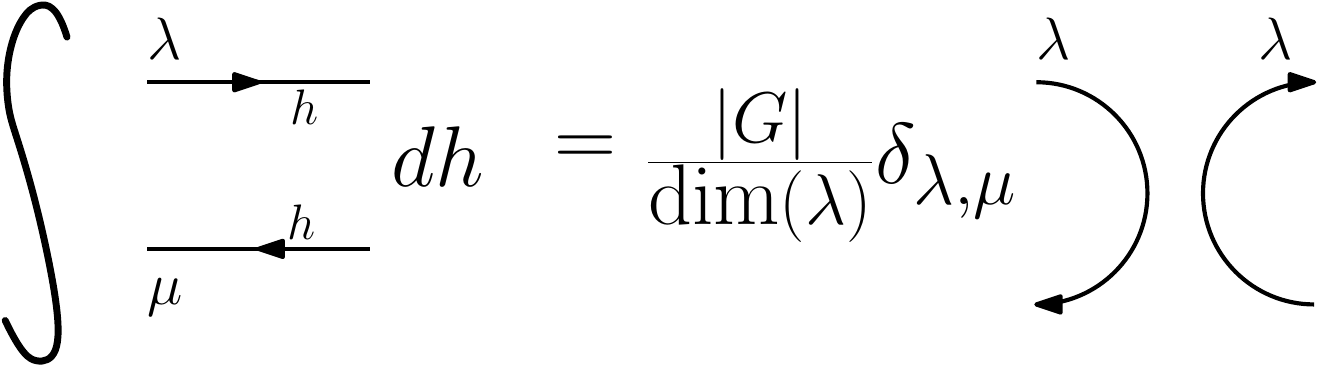}
\caption{Graphical representation of the identity (\ref{integ-2}). }
\end{center}
\end{figure}

\begin{figure}[t!]\label{int-3}
\begin{center}
\includegraphics[width=1.0\textwidth, angle=0.0, scale=0.6]{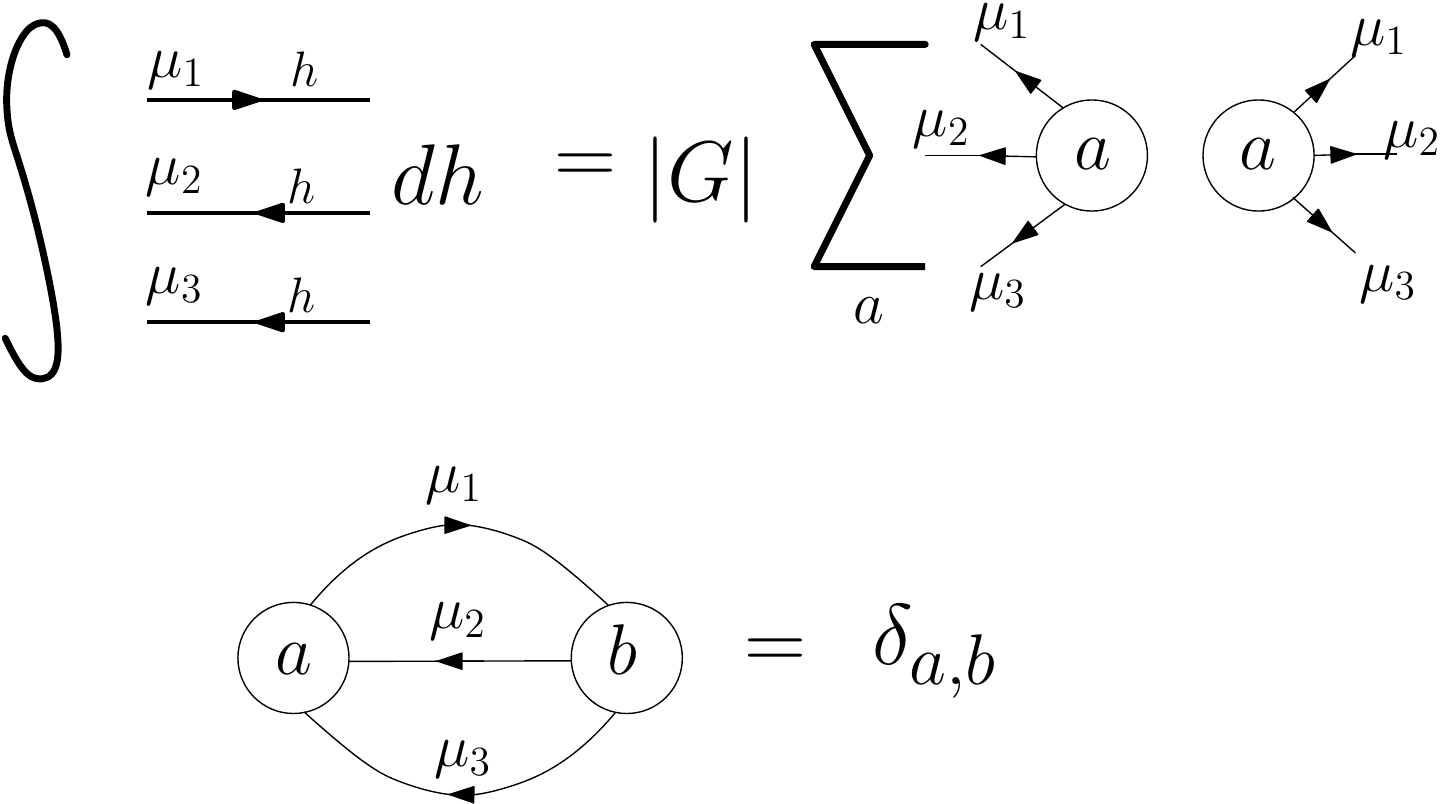}
\caption{ Graphical representation of the identity (\ref{integ-3}). }
\end{center}
\end{figure}

Choose a basis $\{e_i\}$ in the representation space $V_\lambda$ and denote by $\pi^\lambda(g)^i_j$
matrix elements of linear operator $\pi_\lambda(g)$ acting on this space $\pi_\lambda(g)e_j=\sum_i
\pi^\lambda(g)^i_j e_i$. Let $dh$ be the Haar measure on $G$, then

\begin{equation}\label{integ-2}
\int_{G^{\times 2}} \pi^{\lambda}(h^{-1})_j^i \pi^{\mu}(h)^k_ldh=\frac{1}{\dim(\lambda)}\delta_{\lambda, \mu}\delta_l^i\delta_j^k
\end{equation}
Note that if $\{e^i\}$ is the dual basis then $\pi_\lambda(g)^*(e^i)=\sum_j \pi^\lambda(g)^i_j e^j$.
In the basis free form this identity can be written as
\[
\int_{G^{\times 2}} (\pi^{\lambda})^*(h^{-1})\otimes \pi^{\mu}(h)dh=\frac{1}{\dim(\lambda)}\delta_{\lambda, \mu}\iota \circ ev
\]
where $ev: V_\lambda^*\otimes V_\lambda\to \CC$ is the evaluation map
and $\iota: \CC\to V_\lambda^*\otimes V_\lambda$ is the injection map.

\begin{equation}\label{integ-3}
\int_{G^{\times 3}} \pi^{\mu_1}(h^{-1})^i_j \pi^{\mu_2}(h)^k_l \pi^{\mu_3}(h)^s_tdh=\sum_{a\in (V_{\mu_1}^*\otimes V_{\mu_2}\otimes V_{\mu_3})^G } a^{iks}\overline{a}_{jlt}
\end{equation}
where $\{a\}$ is a basis in subspace of $G$-invariant vectors in the triple tensor product $(V_{\mu_1}^*\otimes V_{\mu_2}\otimes V_{\mu_3})^G$ and 
$\{\overline{a}\}$ is the dual basis in the dual vector space. 

These and other similar identities can be written in a basis free form  as
\[
\int_{G^{\times m}} \pi^{\mu_1}(h)\otimes \pi^{\mu_2}(h)\otimes\dots\otimes \pi^{\mu_m}(h)dh=P_0
\]
where $P_0$ is the orthogonal projector onto the subspace $(V_{\mu_1}\otimes V_{\mu_2}\otimes \dots \otimes V_{\mu_m})^G\subset V_{\mu_1}\otimes V_{\mu_2}\otimes \dots \otimes V_{\mu_m}$ of $G$-invariant vectors.

\end{document}